\documentclass[12pt, draftclsnofoot, onecolumn]{IEEEtran}

\usepackage{cite}
\usepackage[textwidth=6.5in, textheight=24.62cm]{geometry}
\usepackage{amsthm, amsmath,amssymb,amsfonts,bbm}
\usepackage{algorithmic}
\usepackage[dvips]{graphicx}
\usepackage{subfigure}
\graphicspath{{graphics/}}
\usepackage{multirow}
\usepackage{filecontents, pgfplots}
\usepackage[algoruled]{algorithm2e}
\usepackage{microtype}
\usepackage{tikz}
\usepackage{textcomp}
\usepackage{xcolor}
\usepackage{verbatim}
\usepackage[binary-units=true]{siunitx}
\DeclareSIUnit{\belmilliwatt}{Bm}
\DeclareSIUnit{\dBm}{\deci\belmilliwatt}
\def\BibTeX{{\rm B\kern-.05em{\sc i\kern-.025em b}\kern-.08em
		T\kern-.1667em\lower.7ex\hbox{E}\kern-.125emX}}

\makeatletter
\newif\iftag@here

\newcommand*{\taghere}[1][0pt]
{\ifmeasuring@\else
	\global\tag@heretrue
	\tikz[remember picture,overlay]{\coordinate (taghere) at (0pt,#1);}%
	\fi}

\def\place@tag{%
	\iftagsleft@
	\kern-\tagshift@
	\iftag@here
	\global\tag@herefalse
	\tikz[remember picture,overlay]%
	{\path (taghere) -| node[anchor=base]{\rlap{\boxz@}} (0pt,0pt);}%
	\else
	\if1\shift@tag\row@\relax
	\rlap{\vbox{%
			\normalbaselines
			\boxz@
			\vbox to\lineht@{}%
			\raise@tag
	}}%
	\else
	\rlap{\boxz@}%
	\fi
	\kern\displaywidth@
	\fi
	\else
	\kern-\tagshift@
	\iftag@here
	\global\tag@herefalse
	\tikz[remember picture,overlay]%
	{\path  (taghere) -|  node[anchor=base]{\llap{\boxz@}} (0pt,0pt);}%
	\else
	\if1\shift@tag\row@\relax
	\llap{\vtop{%
			\raise@tag
			\normalbaselines
			\setbox\@ne\null
			\dp\@ne\lineht@
			\box\@ne
			\boxz@
	}}%
	\else \llap{\boxz@}%
	\fi
	\fi
	\fi
}
\makeatother

\DeclareMathOperator*{\argmax}{arg\,max}
\DeclareMathOperator*{\maximize}{maximize}
\DeclareMathOperator*{\argmin}{arg\,min}

\DeclareMathOperator*{\subjectto}{subject\,to}

\usepackage[acronym,nomain]{glossaries}
\makeglossaries
\newacronym{swipt}{SWIPT}{simultaneous wireless information and power transfer}
\newacronym{wpt}{WPT}{wireless power transfer}
\newacronym{wit}{WIT}{wireless information transfer}
\newacronym{awgn}{AWGN}{additive white Gaussian noise}
\newacronym{tx}{TX}{transmitter}
\newacronym{ir}{IR}{information receiver}
\newacronym{eh}{EH}{energy harvester}
\newacronym{ap}{AP}{average power}
\newacronym{pp}{PP}{peak power}

\newacronym{siso}{SISO}{single-input single-output}
\newacronym{mimo}{MIMO}{multiple-input multiple-output}
\newacronym{miso}{MISO}{multiple-input single-output}
\newacronym{simo}{SIMO}{single-input multiple-output}

\newacronym{rf}{RF}{radio frequency}
\newacronym{dc}{DC}{direct current}
\newacronym{ac}{AC}{alternative current}
\newacronym{papr}{PAPR}{peak-to-average power ratio}
\newacronym{lp}{LPF}{low-pass filter}
\newacronym{mc}{MC}{matching circuit}
\newacronym{mrt}{MRT}{maximum ratio transmission}

\newacronym{rv}{RV}{random variable}
\newacronym{iid}{i.i.d.}{independent and identically distributed}
\newacronym{pdf}{pdf}{probability density function}

\newacronym{dnn}{DNN}{dense neural networks}
\glsunset{dnn}
\newacronym{mdp}{MDP}{Markov decision process}

\newacronym{sca}{SCA}{successive convex approximation}
\newacronym{sdr}{SDR}{semi-definite relaxation}

\newacronym{spr}{LP}{low power}
\newacronym{mpr}{MP}{medium power}
\newacronym{lpr}{HP}{high power}

\DeclareMathOperator{\rank}{rank}
\newcommand{\norm}[1]{\left\lVert#1\right\rVert_2}
\newcommand{\Tr}[1]{\text{Tr}\{#1\} }
\begin{document}

	\newtheorem{proposition}{Proposition}	
	\newtheorem{lemma}{Lemma}	
	\newtheorem{corollary}{Corollary}
	\newtheorem{assumption}{Assumption}	
	\newtheorem{remark}{Remark}	
	
	\title{Optimal Transmit Strategy for Multi-user MIMO WPT Systems With Non-linear Energy Harvesters}
	
	\author{\IEEEauthorblockN{Nikita Shanin, Laura Cottatellucci, and Robert Schober}
	\IEEEauthorblockA{\textit{Friedrich-Alexander-Universit\"{a}t Erlangen-N\"{u}rnberg (FAU), Germany} }	}
	
	\maketitle
	
\begin{abstract}
In this paper, we study multi-user multi-antenna wireless power transfer (WPT) systems, where each antenna at the energy harvesting (EH) nodes is connected to a dedicated non-linear rectifier.
We propose an optimal transmit strategy which maximizes a weighted sum of the average harvested powers at the EH nodes under a constraint on the power budget of the transmitter.
First, for multiple-input single-output (MISO) WPT systems, we show that it is optimal to transmit scalar symbols with an arbitrary phase and an amplitude, whose probability density function (pdf) has at most two mass points, using maximum ratio transmission (MRT) beamforming.
Then, we prove that for single-input multiple-output (SIMO) WPT systems, the optimal transmit symbol amplitudes are discrete random variables, whose pdf also has no more than two mass points.
For general multi-user MIMO WPT systems, we show that the optimal transmit strategy involves scalar unit-norm symbols with arbitrary phase and at most two beamforming vectors.
In order to determine these vectors, we formulate a non-convex optimization problem and obtain an optimal solution based on monotonic optimization.
Since the computational complexity of the optimal solution is high, we also propose a low-complexity iterative algorithm to obtain a suboptimal solution, which achieves nearly optimal performance.
Our simulation results reveal that the proposed transmit strategy for multi-user MIMO WPT systems outperforms two baseline schemes, which are based on a linear EH model and a single beamforming vector, respectively.
For a given transmit power budget, we show that the harvested power saturates when increasing the number of transmit antennas.
Finally, we observe that the harvested power region spanned by multiple EH nodes is convex and the power harvested at one EH node can be traded for a higher harvested power at the other nodes.
\let\thefootnote\relax\footnotetext{This paper was submitted in part to the Biennial Symposium on Communications, Saskatoon, Canada, 2021 \cite{Shanin2021a} and the IEEE 3rd International Workshop on Wirelessly Powered Systems and Networks, 2021 \cite{Shanin2021b}. }
\end{abstract}
\setcounter{footnote}{0}
 
	\section{Introduction}
	\label{Section:Introduction}	
	The device density in wireless communication networks has significantly increased over the past decades.
The current trends for wireless systems suggest that the number of connected devices will continue to grow over the next few years and a longer battery life for these devices is highly desirable \cite{Akyildiz2020}.
However, efficient charging of the batteries of wireless devices remains an unsolved problem.
Since \gls*{rf} signals are capable of transferring power, in recent years, far-field \gls*{wpt} has attracted significant attention \cite{Grover2010, Zhang2013, Clerckx2019, Boshkovska2015, Kim2020, Xiong2017, Ma2019, Clerckx2016a, Clerckx2018, Huang2017, Shen2020, Huang2018, Morsi2019, Shanin2020}.

In \cite{Grover2010}, the authors studied \gls*{siso} WPT systems and showed that the power transferred to the \gls*{eh} is maximized if a single sinusoidal signal is broadcasted by the \gls*{tx}. 
The authors of \cite{Zhang2013} extended these results to \gls*{mimo} WPT systems and showed that the \emph{input power} at the EH is maximized if a scalar input symbol and \emph{energy beamforming}, i.e., beamforming in the direction of the dominant eigenvector of the channel matrix, are employed at the \gls*{tx}.
Although the solutions developed in \cite{Grover2010} and \cite{Zhang2013} are optimal for the maximization of the power received by the EH, they do not necessarily maximize the \emph{harvested power} since practical EH circuits are non-linear \cite{Clerckx2019, Boshkovska2015, Kim2020}.
Hence, an accurate modeling of the EH circuit is crucial for the design of \gls*{wpt} systems \cite{Clerckx2019, Kim2020, Boshkovska2015, Xiong2017, Ma2019, Clerckx2016a, Clerckx2018, Huang2017, Shen2020, Huang2018, Morsi2019, Shanin2020}.

Practical EHs typically employ a rectenna, i.e., an antenna followed by a rectifier circuit that includes a non-linear element, namely, a diode.
The experimental results in \cite{Boshkovska2015} and \cite{Kim2020} showed that rectennas exhibit a non-linear behavior in both the low and high input power regimes.
In particular, for low input power levels, the rectifier non-linearity is caused by the non-linear forward bias current-voltage characteristic of the diode \cite{Horowitz1989}, whereas in the high input power regime, practical EH circuits suffer from saturation due to the breakdown effect of the diode \cite{Tietze2012}.
In order to capture these non-linearities, the authors in \cite{Boshkovska2015} modeled the harvested power as a parameterized sigmoidal function of the received power whose parameters depend on the waveform of the received signal.
The model in \cite{Boshkovska2015} is widely exploited for the design of \gls*{wpt} systems, e.g., \cite{Xiong2017, Ma2019}, where Gaussian distributed input signals are assumed.
In particular, in \cite{Xiong2017}, it was shown that, adopting the EH model from \cite{Boshkovska2015}, the energy beamforming proposed in \cite{Zhang2013} is also optimal for MIMO WPT systems.
Furthermore, to avoid saturation of the rectenna circuits, based on the model in \cite{Boshkovska2015}, the authors in \cite{Ma2019} proposed to split the RF power received at the EH node between several collocated rectifiers.

Although the model in \cite{Boshkovska2015} characterizes the non-linear behavior of rectenna circuits, it is applicable only for signals with a known fixed waveform and does not allow the optimization of the waveform of the transmit signal \cite{Clerckx2019}.
Therefore, the authors in \cite{Clerckx2016a} proposed a non-linear EH model derived from the Taylor series expansion of the current flow through the rectifying diode of the EH. 
Based on this model, the authors in \cite{Clerckx2018} studied a WPT system with multiple antennas at the TX and a single antenna at the EH node, i.e., a \gls*{miso} WPT system, and showed that the harvested power is maximized with energy beamforming \cite{Zhang2013}, which reduces to scaled \gls*{mrt} in this case.
However, in \cite{Huang2017}, it was shown that energy beamforming is not optimal for general MIMO WPT systems employing non-linear rectenna circuits.
In \cite{Shen2020}, the authors considered a MIMO WPT system, where the EH node is equipped with a single rectifier and proposed to combine the received signals  of different antennas in the RF domain to increase the power at the input of the rectifier and, thus, maximize the harvested power.
However, the practical implementation of the RF combining schemes considered in \cite{Ma2019} and \cite{Shen2020} requires complicated circuit designs and may also introduce associated losses, which are not desirable in practical systems \cite{Liang2005}.
Finally, the authors in \cite{Huang2018} considered a \gls*{mimo} WPT system, where each antenna of the EH node was equipped with a dedicated rectifier, and proposed an iterative algorithm to determine the TX beamforming vector that maximizes the weighted sum of powers harvested by the rectifiers.

Although the results in \cite{Clerckx2016a, Clerckx2018, Huang2017, Huang2018, Shen2020} provide important insights for the design of practical MIMO WPT systems, their applicability is limited to low input power levels at the EH since the saturation of the harvested power is neglected in the underlying EH model \cite{Clerckx2016a}.
A realistic EH model that accurately captures the rectenna non-linearity for both low and high input powers was developed in \cite{Morsi2019}.
The analysis in \cite{Morsi2019} showed that for \gls*{siso} WPT systems, it is optimal to adopt ON-OFF signaling at the \gls*{tx}, where the ON symbol and its probability are chosen to maximize the harvested power without saturating the EH while satisfying an average power constraint at the TX.
The optimality of ON-OFF signaling was confirmed in \cite{Shanin2020}, where a learning-based approach was employed to model non-linear rectenna circuits equipped with a single and multiple diodes, respectively \cite{Tietze2012}.
Finally, in \cite{Shanin2021a} and \cite{Shanin2021b}, which are the conference versions of this paper, exploiting the rectenna model derived in \cite{Morsi2019}, the authors studied the harvested power region of a two-user MISO WPT system and the maximum performance of a single-user MIMO WPT system, respectively.
However, to the best of the authors' knowledge, the problem of optimizing the transmit strategy for multi-user MIMO WPT systems, where the EH nodes are equipped with multiple rectennas exhibiting non-linear behavior in both the low and high input power regimes, has not been solved, yet.
We note that the EH model adopted in \cite{Shanin2021a, Shanin2021b} is a special case of the more general EH model considered in this paper.
Hence, the results obtained for two-user \gls*{miso} and single-user \gls*{mimo}  WPT systems in \cite{Shanin2021a} and \cite{Shanin2021b}, respectively, are special cases of the results presented in this paper.

In this paper, we aim at determining the optimal transmit strategy for multi-user MIMO WPT systems, where multiple EH nodes are equipped with multiple non-linear rectennas.
In order to take the non-linearity of the EH into account, we consider a general rectenna model characterized by a set of properties, which are typically satisfied for the practical rectenna circuits considered in the literature \cite{Zhang2013, Clerckx2019, Kim2020, Boshkovska2015, Xiong2017, Ma2019, Clerckx2016a, Clerckx2018, Huang2017, Shen2020, Huang2018, Morsi2019, Shanin2020}.
Where appropriate, we specialize our results to the non-linear EH model proposed in \cite{Morsi2019}.
We propose an optimal transmit strategy, which is characterized by the distribution of the transmit symbol vector that maximizes a weighted sum of the average harvested powers at the EH nodes subject to a constraint on the power budget of the TX.
The main contributions of this paper can be summarized as follows:
\begin{itemize}
	\item For MISO WPT systems, we show that the optimal transmit strategy comprises \gls*{mrt} beamforming and a scalar input symbol with an arbitrary phase and a discrete random amplitude, whose \gls*{pdf} has at most two mass points.
	The optimal \gls*{pdf} of the symbol amplitudes is the solution of an optimization problem, which is solved via a two-dimensional grid search \cite{Coope2001}.	
	Furthermore, for the EH model in \cite{Morsi2019}, we obtain a closed-form solution and show that the optimal pdf reduces to ON-OFF signaling.
	\item For multi-user SIMO WPT systems, we show that the optimal transmit symbol amplitude is a discrete random variable, whose distribution has at most two mass points and can also be obtained by a two-dimensional grid search.
	Similar to the MISO case, we show that for SIMO WPT systems with two rectennas modeled as in \cite{Morsi2019}, the optimal distribution can be obtained in closed-form and ON-OFF signaling is optimal if the power budget of the TX is low. 
	Furthermore, if it is affordable by the average power constraint of the TX, the discrete transmit symbol amplitudes are chosen to drive one or both rectifiers into saturation, respectively.
	Finally, for high average power budgets, the optimal policy is to saturate both rectifiers and the optimal pdf consists of a single mass point.  
	\item For general MIMO WPT systems, we show that the optimal transmit strategy employs a scalar input symbol and at most two beamforming vectors, which can be determined as solution of a non-convex optimization problem.
	The optimal solution of this problem is obtained via monotonic optimization \cite{Zhang2012}.
	\item To reduce the high computational complexity of determining the optimal beamforming vectors, we develop a low-complexity iterative algorithm based on \gls*{sdr} and \gls*{sca} to obtain a suboptimal solution. Our simulation results reveal that although the suboptimal solution for MIMO WPT systems has a much lower computational complexity than the optimal one, both solutions yield a similar performance.
	\item Our simulations show that the proposed MIMO WPT design outperforms two baseline schemes, one based on the linear EH model in \cite{Zhang2013} and the other based on a single beamforming vector at the TX.
	For the multi-user scenario and a given transmit power budget, we observe that the total average harvested power saturates when the TX is equipped with a large number of antennas.
	Finally, we observe that the harvested power region spanned by multiple EH nodes is convex and the average power harvested at one EH node can be traded for a higher harvested power at the other nodes.
\end{itemize}

The remainder of this paper is organized as follows. 
In Section II, we introduce the system model and discuss the adopted EH model.
In Section III, we formulate an optimization problem for the maximization of the weighted sum of the average harvested powers at the EH nodes and establish a preliminary mathematical result needed for solving the problem.
In Section IV, we determine the optimal transmit strategies for MISO, multi-user SIMO, and multi-user MIMO WPT systems, respectively.
In Section V, we provide numerical results to evaluate the performance of the proposed designs.
Finally, in Section VI, we draw some conclusions.

\emph{Notation:} Bold upper case letters $\boldsymbol{X}$ represent matrices and ${X}_{i,j}$ denotes the element of $\boldsymbol{X}$ in row $i$ and column $j$. 
Bold lower case letters $\boldsymbol{x}$ stand for vectors and ${x}_{i}$ is the $i^\text{th}$ element of $\boldsymbol{x}$.
$\boldsymbol{X}^H$, $\Tr{\boldsymbol{X}}$, and $\rank \{\boldsymbol{X}\}$ denote the Hermitian, trace, and rank of matrix $\boldsymbol{X}$, respectively.
The expectation with respect to random variable $x$ is denoted by $\mathbb{E}_x\{\cdot\}$. 
The real part of a complex number is denoted by $\Re\{ \cdot \}$.
$\boldsymbol{x}^\top$ and $\norm{\boldsymbol{x}}$ represent the transpose and L2-norm of $\boldsymbol{x}$, respectively.
The imaginary unit is denoted by $j$.
The sets of real, real non-negative, and complex numbers are denoted by $\mathbb{R}$, $\mathbb{R}_{+}$, and $\mathbb{C}$, respectively.
$\boldsymbol{1}_K$ and $\boldsymbol{0}_K$ represent column vectors comprising $K$ elements, where all elements are equal to $1$ and $0$, respectively.
The Dirac delta function is denoted by $\delta(x)$. 
$f'(x_0)$ denotes the first-order derivative of function $f(x)$ evaluated at point $x = x_0$.

	\section{System Model and Preliminaries}
	\label{Section:SystemModelPreliminaries}
	In this section, we present the MIMO WPT system model and discuss the adopted multi-antenna EH model.
		\subsection{System Model}
		\label{Section:SystemModel}
		\begin{figure}[!t]
	\centering
	\includegraphics[width=0.45\textwidth, draft = false]{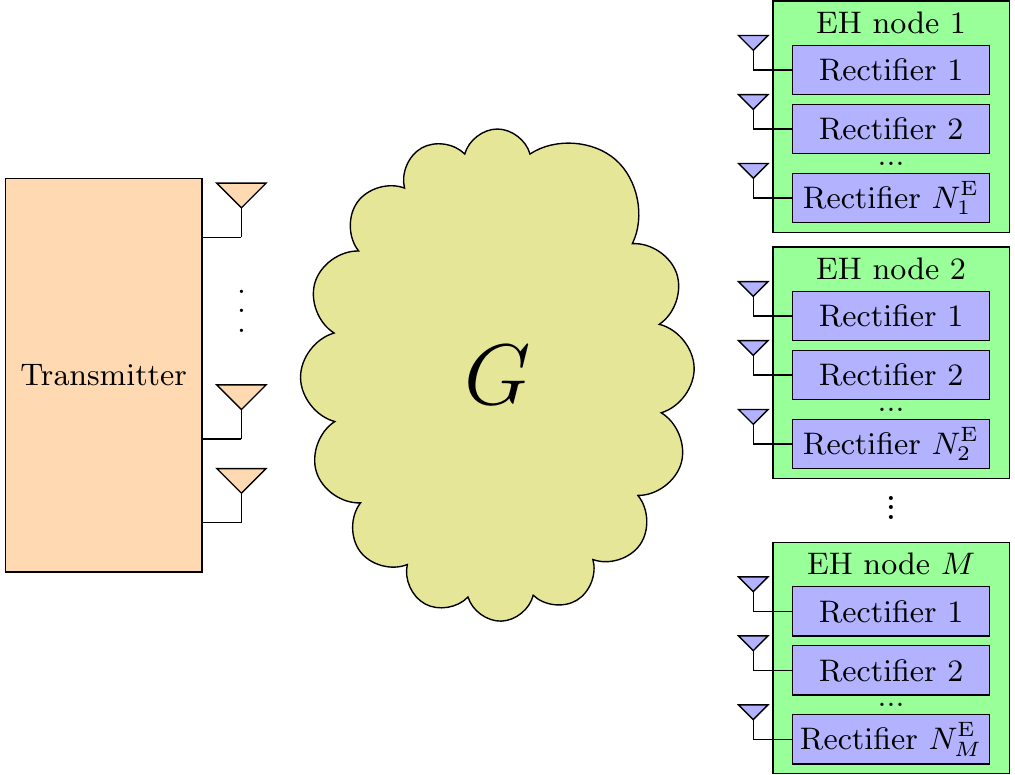}
	\caption{Multi-user MIMO WPT system comprising a multi-antenna TX and $M$ multi-antenna EH nodes, where node $m, m\in\{1,2,\dots,M\}$, is equipped with $N_m^\text{E}$ rectennas. Here, $\boldsymbol{G}$ denotes the channel between the TX and the EH nodes.}
	\label{Fig:SystemModel}
\end{figure}
We consider a narrow-band multi-user MIMO \gls*{wpt} system comprising a \gls*{tx} with $N^\text{T} \geq 1$ antennas and $M$ EH nodes, where EH node $m$, $m \in \{1,2,\dots,M\}$, is equipped with $N^\text{E}_m \geq 1$ antennas, see Fig.~\ref{Fig:SystemModel}.
The TX broadcasts a pulse-modulated RF signal, whose equivalent complex baseband (ECB) representation is modeled as $\boldsymbol{x}(t) = \sum_n \boldsymbol{x}[n] \psi(t-nT)$, where $\boldsymbol{x}[n] \in \mathbb{C}^{N^\text{T}}$ is the transmitted vector in time slot $n$, $\psi(t)$ is the transmit pulse with rectangular shape, and $T$ is the symbol duration.
Transmit vectors $\boldsymbol{x}[n]$ are mutually independent realizations of a random vector $\boldsymbol{x}$, whose \gls*{pdf} is denoted by $p_{\boldsymbol{x}}(\boldsymbol{x})$.

The ECB channel between the TX and antenna $p$ of EH node $m$ is characterized by row-vector $\boldsymbol{g}^m_p \in \mathbb{C}^{1 \times N^\text{T}}$, $p \in \{1,2,\dots,N^\text{E}_m\}$.
Thus, the RF signal received in time slot $n$ at antenna $p$ of EH node $m$ is given by ${z^\text{RF}_{p,m}}(t) = \sqrt{2} \Re\{ \boldsymbol{g}^m_p \boldsymbol{x}(t) \exp(j 2 \pi f_c t) \} $, where $f_c$ denotes the carrier frequency.
The noise received at the EH nodes is ignored since its contribution to the harvested energy is negligible.

		\subsection{Energy Harvester Model}
		\label{Section:EhModel}
		In this paper, we assume that EH node $m$ is equipped with $N^\text{E}_m$ rectennas, i.e., each antenna is connected to a dedicated rectifier, see Fig.~\ref{Fig:SystemModel}.
Thus, the received ECB signal\footnotemark\hspace*{0pt} at rectenna $p \in \{1,2,\dots, N^\text{E}_m \}$ of EH node $m$ in time slot $n$ is given by $z^m_p[n] = \boldsymbol{g}^m_p \boldsymbol{x}[n]$.
\footnotetext{We note that the EH nodes do not convert the RF signal into baseband. However, since the amount of the harvested power can be expressed as a function of the ECB signal received at the EH node, see, e.g., \cite{Morsi2018}, \cite{Clerckx2018}, to simplify the notation, we use the ECB signal representation of the received RF signal.}

Each rectenna comprises an antenna, a matching circuit, a non-linear rectifier with a low-pass filter, and a load resistor \cite{Clerckx2019, Morsi2019, Shanin2020}. 
In order to maximize the power transferred to the rectifier, the matching circuit is typically well-tuned to the carrier frequency $f_c$ and is designed to match the input impedance of the non-linear rectifier circuit with the output impedance of the antenna \cite{LePolozec2016}.
The rectifier is an electrical circuit that comprises a non-linear diode and a low-pass filter to convert the RF signal ${z^\text{RF}_{p,m}}(t)$ received by rectenna $p$ of EH $m$ to a \gls*{dc} signal at the load resistor $R_\text{L}$ of the rectenna.

In this paper, we make the following assumptions concerning the rectenna circuit.
\begin{assumption}	
	\label{Assumption:Memoryless}
	The rectenna circuit is memoryless, i.e., the amount of power harvested in time slot $n$ depends on the ECB signal received in time slot $n$ only.
\end{assumption}
\begin{assumption}
	\label{Assumption:Narrowband}
	The harvested power depends on the magnitude of the received ECB signal only and is independent of its phase.	
\end{assumption}
\begin{assumption}	
	\label{Assumption:IncreasingBounded}
	The mapping between the received ECB signal $z$ and the harvested power $P_\text{L}$ is characterized by a non-linear and monotonically non-decreasing function\footnotemark $P_\text{L} = \phi(|z|^2)$.
\end{assumption}
\footnotetext{In this paper, to simplify the notation, we additionally assume that all electrical circuits equipped at the rectennas of the EH nodes are identical, i.e., all rectennas are characterized by the same function $\phi(\cdot)$. The generalization to the case, where the rectennas employ different electrical circuits, is straightforward.}

Assumption~\ref{Assumption:Memoryless} is justified if the symbol duration $T$ is sufficiently large.
In this case, we can neglect the ripples of the voltage level across the load resistor $R_\text{L}$ and the charging and discharging times of the reactive elements of the circuit.
Thus, we can assume that the output voltage level of rectenna $p$ of EH node $m$ in time slot $n$ is constant and depends only on the signal $z_p^m[n]$ \cite{Clerckx2016, Morsi2019, Shanin2020}.
Assumption~\ref{Assumption:Narrowband} is justified since, for the considered narrow-band signals, the rectenna circuit behaves as an envelope detector \cite{Horowitz1989} and, thus, its behavior is fully characterized by the magnitude $|z_p^m[n]|$ of the received ECB signal $z_p^m[n]$.
Assumption~\ref{Assumption:IncreasingBounded} is satisfied since typical rectenna circuits include a diode that has a non-linear non-decreasing current-voltage characteristic \cite{Tietze2012, Clerckx2016, Morsi2019}.

\textbf{Example:}
In this paper, as an example for an EH model that satisfies Assumptions \ref{Assumption:Memoryless}-\ref{Assumption:IncreasingBounded}, we adopt the model proposed in \cite{Morsi2019}. 
The corresponding power harvested by the rectenna as a function of the magnitude of the received ECB signal $z$ is given as follows:
\begin{equation}
	\phi(|{z}|^2) = \min\big\{ \varphi(|z|^2) , \varphi(A_s^2) \big\},
	\label{Eqn:RaniaModel}
\end{equation}
\noindent where $\varphi(|z|^2) = \bigg[\frac{1}{a} W_0 \bigg(a\exp(a) I_0 \Big(B\sqrt{2|{z}|^2}\Big) \bigg)-1 \bigg]^2 I_s^2 R_L$, $a = \frac{I_s(R_L + R_s)}{\mu V_\text{T}}$, $B = \frac{1}{\mu V_\text{T} \sqrt{\Re \{ 1/Z_a^* \} }}$, and $W_0(\cdot)$ and $I_0(\cdot)$ are the principal branch of the Lambert-W function and the modified Bessel function of the first kind and zero order, respectively.
Here, $Z_a^*$, $V_\text{T}$, $I_s$, $R_s$, and $\mu \in [1,2]$ are parameters of the rectenna circuit, namely, the complex-conjugate of the input impedance of the rectifier circuit, the thermal voltage, the reverse bias saturation current, the series resistance, and the ideality factor of the diode, respectively. 
These parameters depend on the circuit elements and are independent of the received signal.
Finally, since for large input power levels, rectenna circuits are driven into saturation \cite{Tietze2012, Boshkovska2015, Morsi2019, Shanin2020}, the function in (\ref{Eqn:RaniaModel}) is bounded, i.e., $\phi(|z|^2) \leq \phi(A_s^2) , \; \forall z \in \mathbb{C}$, where $A_s$ is the minimum input signal magnitude level at which the output power starts to saturate.

		\section{Problem Formulation and Useful Result}
		\label{Section:ProblemFormulation}
		In this section, we formulate an optimization problem for the maximization of the weighted average harvested power of the considered multi-user MIMO WPT system.
		Then, to obtain a preliminary result needed for solving this problem, we formulate and solve an auxiliary optimization problem, where we maximize the expected value of a function of a one-dimensional random variable under a constraint on its mean value.
		\subsection{Problem Formulation}
		We characterize the transmit strategy via the pdf $p_{\boldsymbol{x}}(\boldsymbol{x})$ of transmit symbol vector $\boldsymbol{x}$.
The objective of the proposed transmit strategy is to maximize the weighted average power harvested at the EH nodes under an average power constraint at the TX.
Thus, we formulate the following optimization problem:
\begin{subequations}
	\begin{align}
	\maximize_{ {p}_{\boldsymbol{x}} } \quad &\overline{\Phi}(p_{\boldsymbol{x}}) 
	\label{Eqn:WPT_GeneralProblem_Obj} \\
	\subjectto \quad & \int_{\boldsymbol{x}} \norm{\boldsymbol{x}}^2 p_{\boldsymbol{x}}(\boldsymbol{x}) d \boldsymbol{x} \leq P_x, \label{Eqn:WPT_GeneralProblem_C1}\\
	&\int_{\boldsymbol{x}} p_{\boldsymbol{x}}(\boldsymbol{x}) d \boldsymbol{x} = 1, \label{Eqn:WPT_GeneralProblem_C2}
	\end{align}
	\label{Eqn:WPT_GeneralProblem}
\end{subequations}
\noindent \hspace*{-5pt}where the objective function is the weighted sum of the average harvested powers at the EH nodes defined as
\begin{equation}
\overline{\Phi}(p_{\boldsymbol{x}}) = \sum_{m = 1}^M \xi_m \mathbb{E}_{\boldsymbol{x}} \big\{\psi_m(\boldsymbol{x})\big\}.
\label{Eqn:EhUtilityFunction}
\end{equation}
Here, $\psi_m(\boldsymbol{x}) = \sum_{p=1}^{N^\text{E}_m} \phi(|\boldsymbol{g}_p^m \boldsymbol{x}|^2)$ is the total power harvested by EH node $m$ and $\xi_m \geq 0, m\in\{1,2,\dots,M\}, \sum_m \xi_m = 1,$ is the weight for EH node $m$ \cite{Shen2020}.
We note that the weights associated with the users allow the TX to control the distribution of the harvested power among the EH nodes.
In particular, if weight $\xi_m$ is increased, the optimal transmit strategy will favor EH node $m$ and increase the average harvested power $\mathbb{E}\{\psi_m(\boldsymbol{x})\}$ at EH node $m$ at the expense of the average harvested powers at the other EH nodes.
Furthermore, we impose constraints (\ref{Eqn:WPT_GeneralProblem_C1}) and (\ref{Eqn:WPT_GeneralProblem_C2}) to limit the transmit power budget at the TX and ensure that $p_{\boldsymbol{x}}(\boldsymbol{x})$ is a valid pdf, respectively.

\begin{remark}
	Optimization problem (\ref{Eqn:WPT_GeneralProblem}) may have an infinite number of solutions.
	In particular, for a general EH model satisfying Assumptions~\ref{Assumption:Memoryless} - \ref{Assumption:IncreasingBounded}, since $\| \boldsymbol{x} \|_2$ and the average harvested power $\overline{\Phi}(\cdot)$ are invariant under phase rotation of the transmit symbol vector $\boldsymbol{x}$, given an optimal pdf $p^*_{\boldsymbol{x}}(\boldsymbol{x})$ and a random phase $\phi_x \in [-\pi, \pi)$ with an arbitrary pdf $p_{\phi_x}(\phi_x)$, the random vector $\boldsymbol{\tilde{x}}=\exp(j\phi_x)\boldsymbol{x}$ is still a solution of (\ref{Eqn:WPT_GeneralProblem}) \cite{Shanin2020}.
	Furthermore, we note that, for bounded $\phi(\cdot)$, such as (\ref{Eqn:RaniaModel}), if affordable by the power budget $P_x$, there may be an infinite number of pdfs that drive all the rectifiers of the EH into saturation while satisfying constraints (\ref{Eqn:WPT_GeneralProblem_C1}) and (\ref{Eqn:WPT_GeneralProblem_C2}).
	Thus, in the following, we determine one pdf $p_{\boldsymbol{x}}^*(\boldsymbol{x})$ that solves (\ref{Eqn:WPT_GeneralProblem}).	
\end{remark}
	
As we will see in Section~\ref{Section:OptimalTransmitStrategies}, the optimal solution of (\ref{Eqn:WPT_GeneralProblem}) leverages the solution of a related auxiliary optimization problem.
In the next subsection, we solve this auxiliary problem, namely, the maximization of the expectation of a non-decreasing function $f(\nu)$ of a scalar random variable $\nu$ under a constraint on the mean value of $\nu$.

			\subsection{Auxiliary Optimization Problem}
			Let us consider the following auxiliary optimization problem:
\begin{equation}
	\maximize_{ {p}_{\nu} } \; \mathbb{E}_{\nu} \{f(\nu)\} \quad \subjectto \; \mathbb{E}_{\nu} \{\nu\} \leq A_{\nu},
	\label{Eqn:GeneralOptimizationProblem}
\end{equation}
\noindent whose solution is the pdf $p_{\nu}(\nu)$ which maximizes the expectation of $f(\nu)$ under a constraint on the mean value of $\nu$.
In order to solve (\ref{Eqn:GeneralOptimizationProblem}), let us first define the slope of the straight line connecting points $\big(\nu_1, f(\nu_1)\big)$ and $\big(\nu_2, f(\nu_2)\big)$, where $\nu_2 > \nu_1$, as follows:
\begin{equation}
	\sigma(\nu_1,\nu_2; f) = \frac{f(\nu_2) - f(\nu_1)}{\nu_2 - \nu_1}.
	\label{Eqn:SlopeFunctionF}
\end{equation}
Then, we establish an upper-bound on $\mathbb{E}_\nu \{f(\nu)\}$.
We note that if $f(\nu)$ is convex (concave), an upper-bound on $\mathbb{E}_{\nu}\{f(\nu)\}$ is given by the Edmundson-Madansky (Jensen's) inequality, e.g., \cite{Dokov2002}. 
However, since we intend to apply this result to the weighted sum of functions $\psi_m(\cdot), m\in\{1,2,\dots,M\},$ in (\ref{Eqn:WPT_GeneralProblem}), which are not necessarily convex or concave, in the following lemma, we extend the Edmundson-Madansky and Jensen's inequalities to arbitrary non-decreasing functions $f(\nu)$ and determine an upper-bound on the expectation of $f(\nu)$.

\begin{lemma}
	Let us consider a non-decreasing function $f(\nu)$ of random variable $\nu$. 
	Then, for a given mean value of $\nu$, $\overline{\nu} = \mathbb{E}_\nu\{\nu\}$, the expectation of $f(\nu)$ is upper-bounded by the following inequality:
	\begin{equation}
		\mathbb{E}_\nu \{f(\nu)\} \leq E_f\big( \overline{\nu} \big),
		\label{Eqn:LemmaUpperBound}
	\end{equation}
	where $E_f( \overline{\nu} ) = \beta f(\nu_1^*) + (1-\beta) f(\nu_2^*)$ and $\beta = \frac{\nu_2^* - \overline{\nu}}{\nu_2^* - \nu_1^*}$. 
	Here, $\nu_1^*$ and $\nu_2^*$ are given by $\nu_1^* = \argmin_{\nu_1 \leq \overline{\nu} } \gamma(\nu_1; f)$, where $\gamma(\nu_1; f) = \max_{\nu_2 \geq \overline{\nu}} \sigma(\nu_1,\nu_2; f)$, and $\nu_2^* = \argmax_{\nu_2 \geq \overline{\nu}} \sigma(\nu_1^*,\nu_2; f)$, respectively.
	Furthermore, inequality (\ref{Eqn:LemmaUpperBound}) holds with equality if the pdf of $\nu$ is given by $p_\nu^*(\nu) = \beta \delta(\nu - \nu_1^*) + (1-\beta) \delta(\nu - \nu_2^*)$.
	\label{Theorem:Lemma1}	
\end{lemma}
\begin{proof}
	Please refer to Appendix \ref{Appendix:LemmaProof}.
\end{proof}

\begin{figure}[!t]
	\centering
	\subfigure[]{
		\includegraphics[width=0.32\textwidth, draft=false]{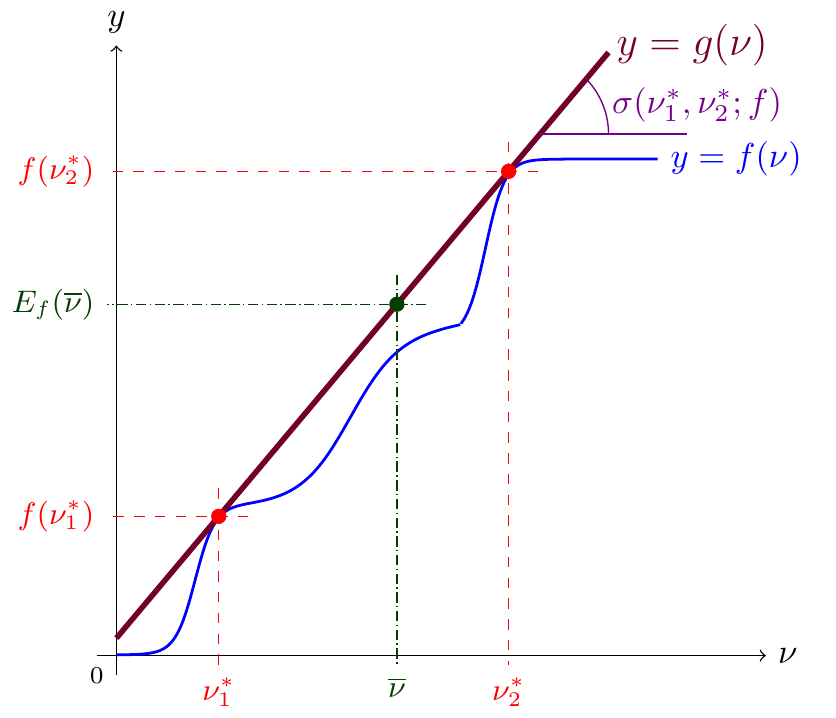}
		\label{Fig:LemmaIllustrationA}
	}
	\quad
	\subfigure[]{
		\includegraphics[width=0.32\textwidth, draft=false]{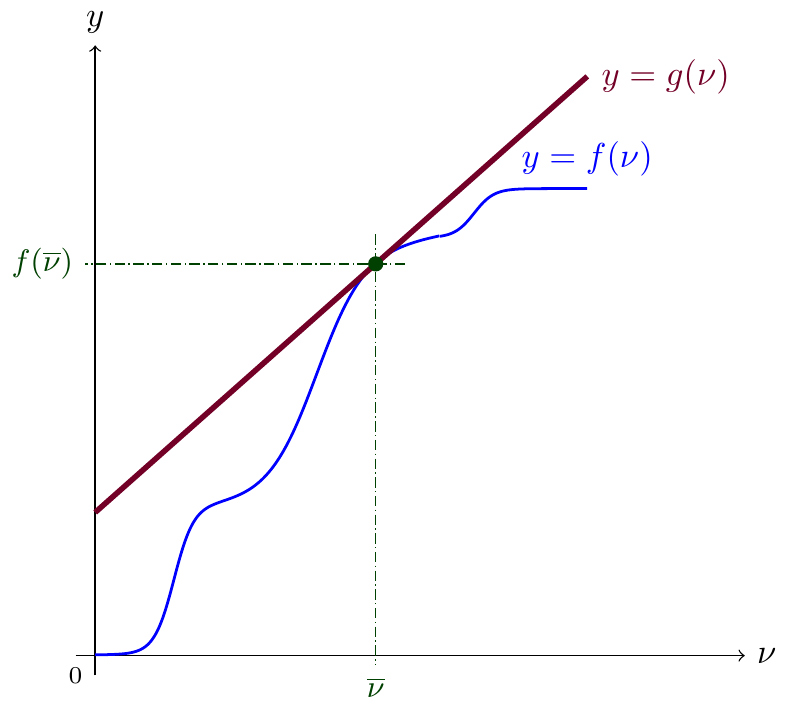}
		\label{Fig:LemmaIllustrationB}
	}	
	\caption{Illustrations of Lemma~\ref{Theorem:Lemma1} and Corollary~\ref{Theorem:Corollary1}.}
	\label{Fig:LemmaIllustration}
\end{figure}
Lemma~\ref{Theorem:Lemma1} is applicable for arbitrary non-decreasing functions $f(\nu)$.
In the following corollary, we show that for a certain class of functions $f(\nu)$, the result in Lemma~\ref{Theorem:Lemma1} can be significantly simplified.
\begin{corollary}
	Let us consider a non-decreasing function $f(\nu)$ of random variable $\nu$. 
	If function $f(\cdot)$ is differentiable at $\overline{\nu} = \mathbb{E}_{\nu} \{\nu\}$ and the following property holds:
	\begin{equation}
		f'(\overline{\nu}) (\overline{\nu}-\nu) \leq f(\overline{\nu}) - f(\nu), \; \forall \nu \in \mathbb{R},
		\label{Eqn:Corollary1_Condition}
	\end{equation}
	\noindent then the expectation of $f(\nu)$ is upper-bounded by $\mathbb{E}_\nu \{f(\nu)\} \leq f\big( \overline{\nu} \big)$, where the inequality holds with equality if the pdf of $\nu$ is given by $p_\nu^*(\nu) = \delta(\nu - \overline{\nu})$.
	\label{Theorem:Corollary1}	
\end{corollary}
\begin{proof}
	Please refer to Appendix \ref{Appendix:Corollary1Proof}.
\end{proof}

We note that $E_f\big( \mathbb{E}_\nu \{\nu\} \big)$ in Lemma~\ref{Theorem:Lemma1} can be interpreted as the value of linear function $g(\nu)$, defined by points $\big(\nu_1^*, f(\nu_1^*)\big)$ and $\big(\nu_2^*, f(\nu_2^*)\big)$, at $\overline{\nu} = \mathbb{E}_\nu \{\nu\}$, where the choice of $\nu_1^*$ and $\nu_2^*$ ensures that $g(\nu) \geq f(\nu)$, $\forall \nu \in \mathbb{R}$, see Fig.~\ref{Fig:LemmaIllustrationA}.
However, if condition (\ref{Eqn:Corollary1_Condition}) in Corollary~\ref{Theorem:Corollary1} is satisfied, points $\nu_1^*$ and $\nu_2^*$ coincide, i.e., $\nu_1^* = \lim_{\epsilon \to 0} \overline{\nu} - \epsilon$ and $\nu_2^* = \lim_{\epsilon \to 0} \overline{\nu} + \epsilon$, and then, $g(\nu) = f(\overline{\nu}) + f'(\overline{\nu}) (\nu-\overline{\nu}) \geq f(\nu)$, $\forall \nu \in \mathbb{R}$, see Fig.~\ref{Fig:LemmaIllustrationB}.

Finally, exploiting Lemma~\ref{Theorem:Lemma1} and Corollary~\ref{Theorem:Corollary1}, in the following corollary, we determine the optimal pdf $p_\nu^*(\nu)$ of random variable $\nu$ solving optimization problem (\ref{Eqn:GeneralOptimizationProblem}).
\begin{corollary}
	A solution\footnotemark\hspace*{1pt} of optimization problem (\ref{Eqn:GeneralOptimizationProblem}) is a discrete pdf given by $p_\nu^*(\nu) = \delta(\nu - A_\nu)$ if condition (\ref{Eqn:Corollary1_Condition}) holds, and $p_\nu^*(\nu) = \beta \delta(\nu - \nu_1^*) + (1-\beta) \delta(\nu - \nu_2^*)$, where
	$\nu_1^* = \argmin_{\nu_1 \leq A_\nu } \gamma(\nu_1; f)$, $\gamma(\nu_1; f) = \max_{\nu_2 \geq A_\nu} \sigma(\nu_1,\nu_2; f)$, $\nu_2^* = \argmax_{\nu_2 \geq A_\nu} \sigma(\nu_1^*,\nu_2; f)$, $\beta = \frac{\nu_2^* - A_\nu}{\nu_2^* - \nu_1^*}$, otherwise.
	\label{Theorem:Corollary2}	
\end{corollary}
\begin{proof}
	Please refer to Appendix \ref{Appendix:Corollary2Proof}.
\end{proof}
\footnotetext{We note that, similar to (\ref{Eqn:WPT_GeneralProblem}), problem (\ref{Eqn:GeneralOptimizationProblem}) may have an infinite number of solutions, i.e., for a given monotonic non-decreasing function $f(\cdot)$, there may exist multiple pdfs $p_{\nu}(\nu)$ that yield the same value of $\mathbb{E}_\nu \{f(\nu)\}$ and satisfy the constraint in (\ref{Eqn:GeneralOptimizationProblem}).
	In Corollary~\ref{Theorem:Corollary2}, we obtain one solution of (\ref{Eqn:GeneralOptimizationProblem}).}

The results in Lemma~\ref{Theorem:Lemma1}, Corollary~\ref{Theorem:Corollary1}, and Corollary~\ref{Theorem:Corollary2} will be exploited in Section~\ref{Section:OptimalTransmitStrategies} for solving the transmit strategy optimization problem in (\ref{Eqn:WPT_GeneralProblem}).

		\section{Optimal Transmit Strategies}	
		\label{Section:OptimalTransmitStrategies}	
		In this section, we first consider MISO and SIMO WPT systems, where the EH and the TX node are equipped with a single antenna, respectively. 
		For each system architecture, we determine the pdf $p^*_{\boldsymbol{x}}(\boldsymbol{x})$ as solution of (\ref{Eqn:WPT_GeneralProblem}) and the resulting optimal transmit strategy.
		Then, we consider the general multi-user MIMO WPT case and present optimal and suboptimal solutions of (\ref{Eqn:WPT_GeneralProblem}).
			\subsection{MISO WPT Systems}
			\label{Section:MISOSystem}
			In the following, we consider MISO WPT systems with a single-antenna EH node, i.e., $M = N^\text{E}_1 = 1$, and a TX equipped with $N^\text{T} \geq 1$ antennas.
In this case, the weighted sum in (\ref{Eqn:EhUtilityFunction}) reduces to $\overline{\Phi}(p_{\boldsymbol{x}}) = \mathbb{E}_{\boldsymbol{x}} \big\{ \phi(|\boldsymbol{g} \boldsymbol{x}|^2) \big\}$, where $\boldsymbol{g}$
is the row-vector representing the channel between the TX and the EH node.
In the following proposition, we provide a solution of optimization problem (\ref{Eqn:WPT_GeneralProblem}) for MISO WPT systems and the corresponding optimal transmit strategy.

\begin{proposition}
	For MISO WPT systems, function $\overline{\Phi}(\cdot)$ is maximized by transmit vectors $\boldsymbol{x} = \boldsymbol{w} s$, where $\boldsymbol{w} = \frac{\boldsymbol{g}^H}{ \norm{\boldsymbol{g}} }$ is the MRT beamforming vector and $s = r_s \exp(j \theta_s)$ is a scalar random symbol with arbitrary phase $\theta_s$ and amplitude $r_s$ following distribution $p^*_{r_s} (r_s)$.
	Furthermore, for the optimal transmit strategy, the pdf of the symbol amplitudes $r_s$ is given by $p^*_{r_s}(r_s) = \delta(r_s - \sqrt{P_x})$ if the following inequality holds:
	\begin{equation}
		{\Phi}'(P_x) \big(P_x - r_s^2\big) \leq {\Phi}({P_x}) - {\Phi}(r_s^2), \; \forall r_s \in \mathbb{R}_{+},
		\label{Eqn:Proposition1Condition}
	\end{equation}
	\noindent where ${\Phi}(r^2) = \phi \big({r^2}{ \norm{\boldsymbol{g}}^2 }\big)$.
	If (\ref{Eqn:Proposition1Condition}) does not hold, the optimal transmit strategy is characterized by the pdf $p^*_{r_s}(r_s) = (1-\beta) \delta(r_s - \sqrt{\nu^*_1}) + \beta \delta(r_s - \sqrt{\nu^*_2})$, where $\beta = \frac{\nu^*_2 - P_x}{\nu^*_2 - \nu^*_1}$.
	Here, $\nu^*_1$ and $\nu^*_2$ are the solutions of the following min-max optimization problem:
	\begin{equation}
		\nu_1^* = \argmin_{\nu_1 \leq P_x } \gamma(\nu_1; \Phi)		
		\label{Eqn:SolutionNu1}
	\end{equation}
	\noindent with $\gamma(\nu_1; \Phi) = \max_{\nu_2 \geq P_x} \sigma(\nu_1, \nu_2; \Phi)$ and $\sigma(\nu_1, \nu_2; \Phi) = \frac{{\Phi}({\nu_2}) - {\Phi}({\nu_1})}{\nu_2 - \nu_1}$ and	
	\begin{equation}
	\nu_2^* = \argmax_{\nu_2 \geq P_x} \sigma (\nu_1^*, \nu_2; \Phi),
	\label{Eqn:SolutionNu2}
	\end{equation}
	\noindent respectively.
	\label{Theorem:Proposition1}
\end{proposition}
\begin{proof}
	Please refer to Appendix~\ref{Appendix:Prop1Proof}.
\end{proof}

Proposition~\ref{Theorem:Proposition1} shows that, as for linear EHs in \cite{Zhang2013}, for the considered non-linear EH model, MRT beamforming is optimal.
Furthermore, similar to the SISO case in \cite{Morsi2019}, for MISO WPT systems, there exists an optimal input symbol amplitude that follows a discrete pdf, $p_{r_s}^*(r_s)$, consisting of at most two mass points.
In particular, it is optimal to adopt a single sinusoidal signal $s$ with amplitude $r_s = \sqrt{P_x}$ and an arbitrary phase $\theta_s$ if condition (\ref{Eqn:Proposition1Condition}) holds.
If (\ref{Eqn:Proposition1Condition}) does not hold, amplitude $r_s$ is a discrete binary random variable.
In this case, in order to obtain the pdf $p^*_{r_s}(r_s)$, the non-convex min-max optimization problem defined by (\ref{Eqn:SolutionNu1}), (\ref{Eqn:SolutionNu2}) has to be solved. 
Due to the low dimensionality of the problem, we propose to obtain the optimal solution via a two-dimensional grid search \cite{Coope2001}.

				\subsubsection{Grid Search Method}
				\label{Section:GridSearch}
				In the following, we propose a grid-search based method for solving the min-max optimization problem in Proposition~\ref{Theorem:Proposition1}.
We note that this problem is not convex since function $\sigma(\nu_1, \nu_2; \Phi)$ is not convex and not concave in $\nu_1$ and $\nu_2$, respectively. 
However, since the dimensionality of the problem is low, performing a grid search to determine $\nu_1^*$ and $\nu_2^*$ entails limited and affordable complexity \cite{Coope2001}.
To this end, we define a uniform grid $\mathcal{P} = \{\rho_0, \rho_1, \rho_2, \dots, \rho_{N_\rho}\}$, where $\rho_0 = 0$, $\rho_j = \Delta_\rho + \rho_{j-1}$, $j = 1,2,\dots,{N_\rho}$, $N_\rho$ is the grid size, and $\Delta_\rho$ is a predefined step size.
Then, we define the smallest element of $\mathcal{P}$ which is larger than $P_x$ as $\rho_n$, i.e., $\rho_n = \min\{\rho_j | \rho_j \geq P_x, j = 0,1,\dots,N_\rho\}$.
Next, we define a matrix $\boldsymbol{S} \in \mathbb{R}^{n\times (N_\rho-n+1)}$, whose elements are the values of function $\sigma(\cdot, \cdot; \Phi)$ evaluated at the elements of $\mathcal{P}$, i.e., ${S}_{i,j} = \sigma(\rho_i, \rho_{j'}; \Phi)$, $i = 0,1,\dots,{n-1}$, $j = {j'}-n$, and ${j'} = n,n+1,\dots,N_\rho$.
Finally, we obtain the power values $\nu_1^* = \rho_{i^*}$ and $\nu_2^*=\rho_{n+j^*}$, where $i^* = \argmin_{i} \max_{j} {S}_{i,j}$ and $j^* = \argmax_{j} {S}_{i^*,j}$, respectively.
The proposed grid-search method is summarized in Algorithm~\ref{AlgorithmGridSearch}.
The computational complexity of the proposed scheme is quadratic with respect to the grid size $N_\rho$ and does not depend on the number of antennas.
\begin{algorithm}[!t]		
	\small				
	\linespread{1.45}\selectfont
	\SetAlgoNoLine%
	\SetKwFor{Foreach}{for each}{do}{end}		
	Initialize: Grid size $N_{\rho}$, step size $\Delta_{\rho}$, maximum TX power $P_x$, initial value $\rho_0=0$.	\\	
	1. Compute the grid $\mathcal{P}$ and the values of $\Phi(\cdot)$ for the grid elements:\\
	\For{$m = 0$ {\upshape to} $N_\rho$}{
		1.1. Compute $\Phi_{m}$ = $\Phi(\rho_{m})$\\
		1.2. Set $\rho_{m+1} = \rho_{m} + \Delta_{\rho}$ \\
	}
	2. Determine grid element $\rho_n = \min\{\rho_j | \rho_j \geq P_x, j = 0,1,\dots,N_\rho \}$ \\
	3. Calculate the elements of matrix $\boldsymbol{S}$ as ${S}_{i,j} = \sigma(\rho_i, \rho_{j'}; \Phi) = \frac{\Phi_{j'} - \Phi_i}{\rho_{j'} - \rho_i}$, $i = 0,1,\dots,{n-1}$, $j = {j'}-n$, and ${j'} = n,n+1,\dots,N_\rho$ \\
	4. Determine power values $\nu_1^* = \rho_{i^*}$ and $\nu_2^*=\rho_{n+j^*}$, where $i^* = \argmin_{i} \max_{j} {S}_{i,j}$ and $j^* = \argmax_{j} {S}_{i^*,j}$ \\
	\textbf{Output:} Optimal values $\nu_1^*$, $\nu_2^*$, $\beta = \frac{\nu_2^* - P_x}{\nu_2^* - \nu_1^*}$
	\caption{\strut Grid search for determining the optimal values $\nu_1^*$, $\nu_2^*$, and $\beta$ }
	\label{AlgorithmGridSearch}
\end{algorithm}	 
				\subsubsection{Special Case}
				\label{Section:MISOSpecialCase}
				In the following, we consider a special case of MISO WPT systems, where the EH model satisfies the following additional assumption.
\begin{assumption}
	For the EH model $\phi(\cdot)$, there is a value $A_s$, such that $\forall z: |z| \leq A_s$, $\phi(|z|^2)$ is a convex function and $\forall z: |z| > A_s$, $\phi(|z|^2) = \phi(A_s^2)$.
	\label{Assumption:ConvexityBoundness}
\end{assumption}
In particular, we note that Assumption~\ref{Assumption:ConvexityBoundness} holds for the EH model in (\ref{Eqn:RaniaModel}).
In the following corollary, we show that for an EH model satisfying Assumption~\ref{Assumption:ConvexityBoundness}, the result in Proposition~\ref{Theorem:Proposition1} can be significantly simplified.
\begin{corollary}
	For an EH model satisfying Assumption~\ref{Assumption:ConvexityBoundness}, the pdf $p_{r_s}^*(r_s)$ of the transmit symbol amplitudes $r_s$ in Proposition~\ref{Theorem:Proposition1} is given by $p^*_{r_s}(r_s) = \delta(r_s - \sqrt{P_x^\text{\normalfont max}})$ if $P_x \geq P_x^\text{\normalfont max} = \frac{A_s^2}{\norm{\boldsymbol{g}}^2}$ and $p^*_{r_s}(r_s) = (1-\beta) \delta(r_s) + \beta \delta(r_s - \sqrt{P_x^\text{\normalfont max}})$, where $\beta = \frac{P_x^\text{\normalfont max} - P_x}{P_x^\text{\normalfont max}}$, otherwise.
	\label{Theorem:Corollary3}
\end{corollary}
\begin{proof}
	First, we note that $\forall \nu \in \mathbb{R}_{+}$, $\Phi(\nu) \leq \Phi({P_x^\text{max}}) = \phi(A_s^2)$.
	Therefore, for $P_x \geq P_x^{\text {\normalfont max}}$, the optimal transmit strategy is characterized by the pdf $p^*_{r_s}(r_s) = \delta(r_s - \sqrt{P_x^\text{\normalfont max}})$ .
	Furthermore, functions $\phi(|z|^2)$ in (\ref{Eqn:RaniaModel}) and ${\Phi}(\nu)$ in Proposition~\ref{Theorem:Proposition1} are convex and increasing in the intervals $|z|^2 \in [0, A_s^2]$ and $\nu \in [0, {P_x^\text{max}}]$, respectively. 
	Hence, $\forall P_x: P_x < P_x^\text{max}$, due to the convexity of $\Phi(\nu)$, the solutions of the optimization problems in Proposition~\ref{Theorem:Proposition1} are given by $\nu_1^* = 0$ and $\nu_2^* = P_x^\text{max}$, respectively, with $\beta = \frac{P_x^\text{\normalfont max} - P_x}{P_x^\text{\normalfont max}}$. 
	This concludes the proof.
\end{proof}

Corollary~\ref{Theorem:Corollary3} reveals that if $P_x < P_x^\text{max}$,  ON-OFF signaling with MRT beamforming is optimal, which is similar to the result obtained for the SISO WPT systems in \cite{Morsi2019}.
Furthermore, for $P_x \geq P_x^\text{max}$, it is affordable to drive the EH node into saturation and, hence, the optimal pdf of ${r_s}$ consists of a single mass point. 
We note that in contrast to Proposition~\ref{Theorem:Proposition1}, where all the power budget $P_x$ is utilized for the transmission, Corollary~\ref{Theorem:Corollary3} shows that for EH models satisfying Assumption~\ref{Assumption:ConvexityBoundness} and $P_x \geq P_x^\text{max}$, there is an optimal transmit strategy, where the average transmit power is equal to $P_x^\text{max}$.

			\subsection{SIMO WPT Systems}
			\label{Section:SIMOSystem}
			In the following, we consider a SIMO WPT system, where $M$ EH nodes are equipped with $N^\text{E}_m \geq 1$, $m \in \{1,2,\dots,M\}$, antennas and the TX has a single antenna, i.e., $N^\text{T} = 1$.
In this case, as in \cite{Shanin2020, Morsi2019}, due to Assumption~\ref{Assumption:Narrowband}, the powers harvested at the rectennas depend on the magnitude of the scalar transmit symbol but not on its phase.
Hence, the weighted sum in (\ref{Eqn:EhUtilityFunction}) can be expressed as a function of the pdf $p_{r_x}$ of the transmit symbol amplitude $r_x = |x|$ as follows $\overline{\Phi}(p_{r_x}) = \mathbb{E}_{r_x} \big\{ \Phi(r_x^2) \big\}$, where 
\begin{equation}
	\Phi(r_x^2) = \sum_{m=1}^{M}  \sum_{p=1}^{N^\text{E}_m} \xi_m \phi (r_x^2 |g_p^m|^2)\label{Eqn:SIMOPowerFunction}
\end{equation} 
\noindent and $|g_p^m|$ is the magnitude of the scalar channel coefficient $g_p^m$ between the transmit antenna and antenna $p$ of EH node $m$.
In the following proposition, we provide a solution of optimization problem (\ref{Eqn:WPT_GeneralProblem}) for SIMO WPT systems and the corresponding optimal transmit strategy.

\begin{proposition}
	For the considered SIMO WPT system, function $\overline{\Phi}(\cdot)$ is maximized for discrete transmit symbol amplitudes following distribution $p_{r_x}^*(r_x)$.
	In particular, the optimal transmit strategy is characterized by the pdf $p^*_{r_x}(r_x) = \delta(r_x - \sqrt{P_x})$, if function $\Phi(r_x^2)$ is differentiable at $r_x^2 = P_x$ and the following inequality holds:
	\begin{equation}
		{\Phi}'(P_x) \big(P_x - r_x^2\big) \leq {\Phi}({P_x}) - {\Phi}(r_x^2), \; \forall r_x \in \mathbb{R}_{+}.
	\label{Eqn:Proposition2Condition}
	\end{equation}
	Furthermore, if (\ref{Eqn:Proposition2Condition}) does not hold, the optimal pdf is given by $p^*_{r_x}(r_x) = \beta \delta(r_x - \sqrt{\nu^*_1}) + (1-\beta) \delta(r_x - \sqrt{\nu^*_2})$, where $\beta = \frac{\nu^*_2 - P_x}{\nu^*_2 - \nu^*_1}$.
	Here, $\nu^*_1$ and $\nu^*_2$ are the corresponding solutions of the optimization problems in (\ref{Eqn:SolutionNu1}) and (\ref{Eqn:SolutionNu2}), respectively, where function $\Phi(\cdot)$ is given by (\ref{Eqn:SIMOPowerFunction}).
	\label{Theorem:Proposition2}
\end{proposition}
\begin{proof}
	Please refer to Appendix~\ref{Appendix:Prop2Proof}.
\end{proof}

Proposition~\ref{Theorem:Proposition2} reveals that there exists an optimal pdf $p^*_{r_x}(r_x)$ of the symbol amplitudes $r_x$ that is discrete and consists of one or two mass points.
In particular, as for SISO and MISO WPT systems in \cite{Grover2010} and in Section \ref{Section:MISOSystem}, respectively, it is optimal to transmit a single sinusoid if condition (\ref{Eqn:Proposition2Condition}) holds.
If (\ref{Eqn:Proposition2Condition}) does not hold, this optimal pdf consists of two mass points, $\nu^*_1$ and $\nu^*_2$, which are obtained as solutions of the min-max optimization problem (\ref{Eqn:SolutionNu1}), (\ref{Eqn:SolutionNu2}).
Due to its low dimensionality, this problem also can be efficiently solved via a two-dimensional grid search, as discussed in Section~\ref{Section:GridSearch} and summarized in Algorithm~\ref{AlgorithmGridSearch}.

			{\itshape Special case:} In the following, we consider a special case of the SIMO WPT system, where the EH model satisfies the following additional assumption.
\begin{assumption}
	For function $\phi(\cdot)$, Assumption~\ref{Assumption:ConvexityBoundness} is satisfied.
	Furthermore, for the value $A_s$, $\forall z: |z| \leq A_s$, the following inequality holds: $\phi(|z|^2) \geq \phi(A_s^2) \Big(\frac{|z|^2}{A_s^2}\Big)^2$.
	\label{Assumption:Quadratic}
\end{assumption}
We note that the condition in Assumption~\ref{Assumption:Quadratic} implies that the power harvested by the rectifier grows slower than quadratically with the input power.
In particular, it can be shown that Assumption~\ref{Assumption:Quadratic} is satisfied for the EH model in (\ref{Eqn:RaniaModel}).
Let us now consider a SIMO WPT system with two rectifiers, i.e.,  $M = 1$ and $N^\text{E}_1 = 2$ or $M= 2$ and $N^\text{E}_1 =N^\text{E}_2 = 1$.
In this case, without loss of generality, we denote the scalar channel coefficients between the transmit antenna and the antennas of the rectifiers by $g_1$ and $g_2$ and assume that $|g_1| \geq |g_2|$.
Then, in the following corollary, we provide a closed-form solution for the optimal $\nu_1^*$ and $\nu_2^*$ in Proposition~\ref{Theorem:Proposition2}.
\begin{corollary}
	Let us consider a SIMO WPT system with two rectifiers and an EH model satisfying Assumption~\ref{Assumption:Quadratic}.
	In this case, if $P_x < \rho_{\text{\upshape{min}}} = \frac{A_s^2}{|g_1|^2}$, the optimal transmit strategy is characterized by the pdf of the transmit symbol amplitudes given by $p^*_{r_x}(r_x) = \beta \delta(r_x) + (1-\beta) \delta(r_x - \sqrt{\rho_\text{\upshape{min}}} )$, where $\beta = \frac{{\rho_\text{\upshape{min}}} - P_x}{\rho_\text{\upshape{min}}}$.
	Furthermore, if $P_x \in [\rho_{\text{\upshape{min}}}, \rho_{\text{\upshape{max}}}) $, where $\rho_{\text{\upshape{max}}} = \frac{A_s^2}{|g_2|^2}$, then $p^*_{r_x}(r_x) = \beta \delta(r_x - \sqrt{\rho_\text{\upshape{min}}} ) + (1-\beta) \delta(r_x - \sqrt{\rho_\text{\upshape{max}}} )$ with $\beta = \frac{\rho_\text{\upshape{max}} - P_x}{\rho_\text{\upshape{max}} - \rho_\text{\upshape{min}}}$.
	Finally, if $P_x \geq \rho_{\text{\upshape{max}}}$, the pdf is given by $p^*_{r_x}(r_x) = \delta(r_x - \sqrt{\rho_\text{\upshape{max}}})$.
	\label{Theorem:Corollary4}
\end{corollary}
\begin{proof}
	Please refer to Appendix~\ref{Appendix:Corollary4Proof}.
\end{proof}

Thus, for SIMO WPT systems with two rectifiers, if the transmit power budget is low, i.e., $P_x < \rho_{\text{min}}$, similar to the SISO and MISO WPT cases, ON-OFF signaling is optimal, where the ON signal drives the rectifier with the best channel conditions into saturation.
Furthermore, for $P_x \in [\rho_{\text{\upshape{min}}}, \rho_{\text{\upshape{max}}})$, the pdf $p^*_{r_x}(r_x)$ has two mass points, which are chosen to drive one and both rectifiers into saturation, respectively.
For $P_x \geq \rho_{\text{\upshape{max}}}$, it is affordable to drive both rectifiers into saturation and, hence, the optimal pdf consists of a single mass point.
Moreover, as for MISO WPT systems, Corollary~\ref{Theorem:Corollary4} reveals that for $P_x \geq \rho_{\text{\upshape{max}}}$, the average transmit power of the optimal transmit strategy is equal to $\rho_{\text{\upshape{max}}}$.

			\subsection{MIMO WPT Systems}
			\label{Section:MIMOSystem}
			In the following, we consider the general multi-user MIMO WPT system in Fig.~\ref{Fig:SystemModel}, where $N^\text{T} \geq 1$ and $N^\text{E}_m \geq 1$ antennas are employed at the TX and EH node $m$, $m \in \{1,2,\dots,M\}$, respectively.
In the following proposition, we characterize a solution of (\ref{Eqn:WPT_GeneralProblem}) and the corresponding optimal transmit strategy.

\begin{proposition}
	For multi-user MIMO WPT systems, function $\overline{\Phi} (\cdot)$ is maximized for discrete random transmit symbol vectors $\boldsymbol{x} = \boldsymbol{w} s$, where $s = \exp(j \theta_s)$ is a unit-norm symbol with an arbitrary phase $\theta_s$.
	Here, $\boldsymbol{w}$ is a discrete random beamforming vector, whose pdf is given by $p^*_{\boldsymbol{w}}(\boldsymbol{w}) = \beta \delta(\boldsymbol{w} - \boldsymbol{w}^*_1) + (1-\beta) \delta(\boldsymbol{w} - \boldsymbol{w}^*_2)$.
	The beamforming vectors $\boldsymbol{w}^*_n, n \in \{1,2\},$ are given by
	\begin{align}
	\boldsymbol{w}^*_n &\in \{ \boldsymbol{w} : \Psi(\boldsymbol{w}) = \Phi(\nu_n^*) \},
	\label{Eqn:MimoPropositionBeamformerProblem} \\
	\Phi(\nu) &= \max_{ \{ \boldsymbol{w} \, | \, \boldsymbol{w} \in \mathbb{C}^{N^\text{\upshape T}}\hspace*{-3pt}, \; \norm{\boldsymbol{w}}^2 = \nu \} } \Psi(\boldsymbol{w}),
	\label{Eqn:MimoPropositionFunction}
	\end{align}
	where $\Psi(\boldsymbol{x}) = \sum_m \xi_m \psi_m(\boldsymbol{x}).$ 
	Here, $\nu^*_1$ and $\nu^*_2$ are the corresponding solutions of the optimization problems in (\ref{Eqn:SolutionNu1}) and (\ref{Eqn:SolutionNu2}), respectively, where $\Phi(\cdot)$ is given by (\ref{Eqn:MimoPropositionFunction}).
	Furthermore, if the following inequality holds:
	\begin{align}
	\Phi'(P_x)(P_x - \nu) \leq \Phi(P_x) - \Phi(\nu), \forall \nu \in \mathbb{R}_{+},
	\label{Eqn:MimoPropositionInequality} 
	\end{align}
	the optimal points $\nu_1^*$ and $\nu_2^*$ coincide and the optimal pdf is given by $p^*_{\boldsymbol{w}}(\boldsymbol{w}) = \delta(\boldsymbol{w} - \boldsymbol{w}^*)$, where $\boldsymbol{w}^* = \boldsymbol{w}^*_1 = \boldsymbol{w}^*_2$.
	\label{Theorem:Proposition3}
\end{proposition}
\begin{proof}
	Please refer to Appendix~\ref{Appendix:Prop3Proof}.
\end{proof}

Proposition~\ref{Theorem:Proposition3} reveals that there is an optimal transmit vector $\boldsymbol{x}$ which is discrete and characterized by scalar unit-norm symbols $s$ with an arbitrary phase\footnotemark\hspace*{0pt} and at most two beamforming vectors, $\boldsymbol{w}_1^*$ and $\boldsymbol{w}_2^*$.\hspace*{-5pt} 
\footnotetext{We note that the phase $\theta_s$ of scalar symbol $s$ can be chosen arbitrarily in each time slot $n$. This degree of freedom can be further exploited, for example, for information transmission \cite{Shanin2020}.}
As in Corollary~\ref{Theorem:Corollary1}, these beamforming vectors coincide if inequality (\ref{Eqn:MimoPropositionInequality}) holds.
We note that $\boldsymbol{w}_1^*$ and $\boldsymbol{w}_2^*$ are characterized by the values $\nu_1^*$ and $\nu_2^*$ obtained as solutions of the non-convex problems (\ref{Eqn:SolutionNu1}) and (\ref{Eqn:SolutionNu2}), respectively, that also can be solved using the grid-search method in Algorithm~\ref{AlgorithmGridSearch}.
However, unlike for MISO and SIMO systems, in order to obtain the value of $\Phi(\cdot)$ for a given transmit power value $\nu$, the maximum value of $\Psi(\cdot)$ as solution of (\ref{Eqn:MimoPropositionFunction}) is required.
We note that (\ref{Eqn:MimoPropositionFunction}) is a non-convex problem and, hence, obtaining its optimal solution is, in general, NP-hard.
However, since problem (\ref{Eqn:MimoPropositionFunction}) belongs to the class of monotonic optimization problems, in Section~\ref{Section:OptimalSolution}, we first obtain the optimal solution exploiting the polyblock outer optimization approach \cite{Zhang2012}.
Then, in Section~\ref{Section:SuboptimalSolution}, for practical EH models satisfying Assumption~\ref{Assumption:ConvexityBoundness}, we propose an iterative low-complexity algorithm to obtain a suboptimal solution of the problem.
			
			\subsubsection{Optimal Solution}
			\label{Section:OptimalSolution}
			\begin{figure*}[!t]
	\centering
	\subfigure[]{
		\includegraphics[width=0.31\textwidth, draft=false]{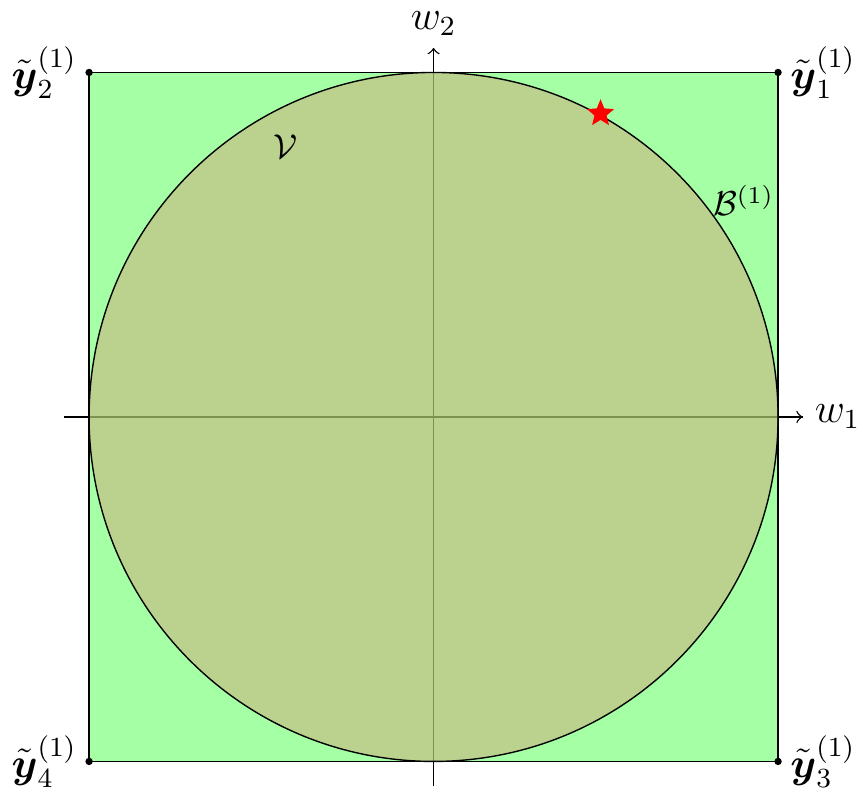}
		\label{Fig:MonotonicOptimizationA}
	}
	\subfigure[]{
		\includegraphics[width=0.31\textwidth, draft=false]{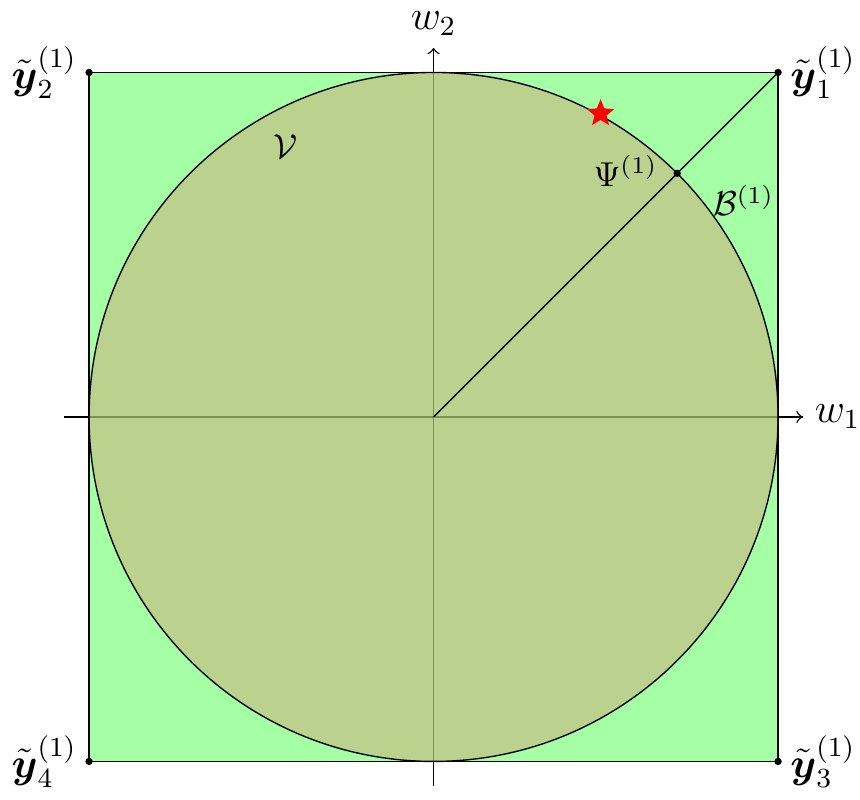}
		\label{Fig:MonotonicOptimizationB}
	}\subfigure[]{
		\includegraphics[width=0.31\textwidth, draft=false]{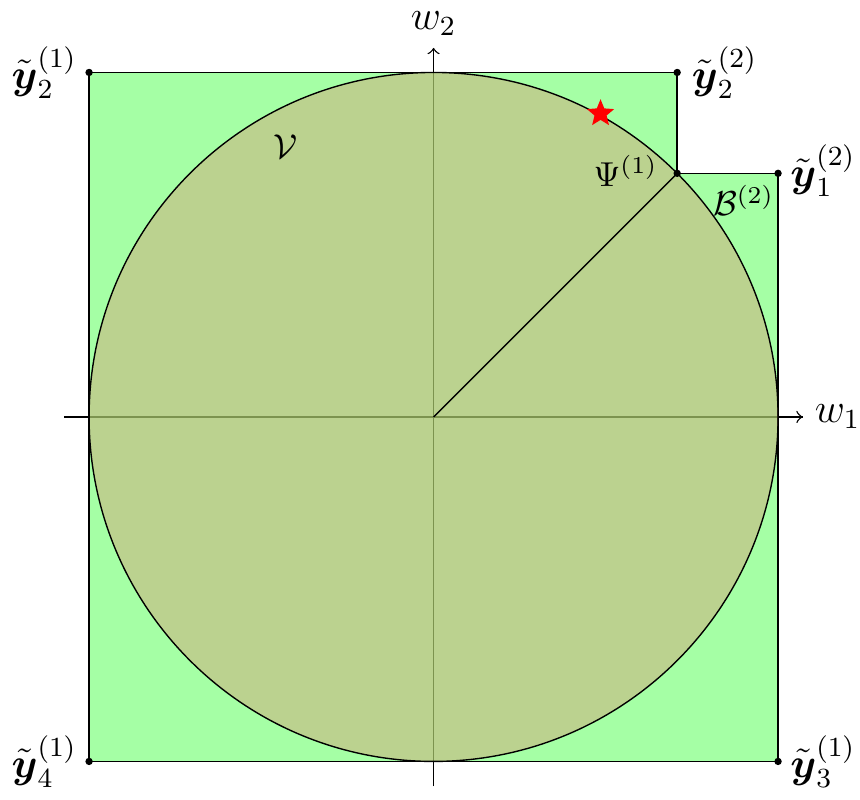}
		\label{Fig:MonotonicOptimizationC}
	}
	\caption{Illustration of the polyblock optimization approach. The red star is the optimal point on the feasible set $\mathcal{V}$.}
	\label{Fig:MonotonicOptimization}
\end{figure*}
In the following, we obtain the optimal solution of non-convex problem (\ref{Eqn:MimoPropositionFunction}) exploiting monotonic optimization \cite{Zhang2012}.
To this end, similar to the monotonic polyblock optimization framework in \cite{Zhang2012}, \cite{Ghanem2020}, we obtain the optimal solution of (\ref{Eqn:MimoPropositionFunction}) by exploring the feasible set $\mathcal{V}$ of (\ref{Eqn:MimoPropositionFunction}) determined by constraint $\norm{\boldsymbol{w}}^2 = \nu$.
First, we enclose $\mathcal{V}$ by constructing an initial polyblock $\mathcal{B}^{(1)}$ with an initial set of vertices $\mathcal{L}^{(1)} = \big\{\tilde{\boldsymbol{y}}_d^{(1)} \vert d \in \{1,2,\dots, 2^{2N^\text{T}} \} \big\}$, where $\tilde{\boldsymbol{y}}_d^{(1)} \in \mathbb{C}^{N^\text{T}}$ is a vector, whose $k^\text{th}$ element is defined as $\tilde{y}_{d,k}^{(1)} = (-1)^{a_{n-1}^d} \nu + j (-1)^{a_{n}^d} \nu$.
Here, $n = 2^k$ and $a_{p}^d \in \{0,1\}$ denotes bit $p$ in the binary representation of number $d$, see Fig.~\ref{Fig:MonotonicOptimizationA}.
Then, since the objective function in (\ref{Eqn:MimoPropositionFunction}) is monotonically increasing in $|\boldsymbol{w}|$, i.e., in $|{w}_1|, |{w}_2|, \dots, |{w}_{N^T}|$, in iteration $m \geq 1$ of the proposed algorithm, as shown in Fig.~\ref{Fig:MonotonicOptimizationB}, we choose a vertex $\boldsymbol{y}^{(m)}$ from the set of vertices $\mathcal{L}^{(m)}$ that maximizes the objective function, i.e.,  $\boldsymbol{y}^{(m)} = \argmax_{\boldsymbol{y} \in \mathcal{L}^{(m)} } \Psi(\boldsymbol{y})$.
We calculate the intersection point $\psi^{(m)}$ between the feasible set $\mathcal{V}$ and the line that connects the origin and vertex $\boldsymbol{y}^{(m)}$ as $\boldsymbol{\chi}^{(m)} = \sqrt{\nu} \frac{\boldsymbol{y}^{(m)}}{\norm{\boldsymbol{y}^{(m)}}} $.
Then, based on vertex set $\mathcal{L}^{(m)}$, we generate a set of new vertices $\mathcal{L}_\mathcal{B}^{(m+1)} = \{\tilde{\boldsymbol{y}}^{(m+1)}_1, \tilde{\boldsymbol{y}}^{(m+1)}_2, \dots, \tilde{\boldsymbol{y}}^{(m+1)}_{N^\text{T}}\}$ and, thus, construct the new polybock $\mathcal{B}^{(m+1)}$ with vertex set $\mathcal{L}^{(m+1)} = \mathcal{L}^{(m)} \cup \mathcal{L}_\mathcal{B}^{(m+1)} \setminus \{ \boldsymbol{y}^{(m)}\}$, see Fig.~\ref{Fig:MonotonicOptimizationC}.
This procedure is continued until the feasible set $\mathcal{V}$ is enclosed by the final polyblock $\mathcal{B}^{(1)} \supset \mathcal{B}^{(2)} \supset \dots \supset \mathcal{V}$.
Finally, as solution of (\ref{Eqn:MimoPropositionFunction}), we select the vertex $\boldsymbol{\chi}^{*}$ that maximizes the objective function $\Psi(\boldsymbol{w})$.
The proposed algorithm is summarized in Algorithm~\ref{OptimalAlgorithm}.

\begin{algorithm}[!t]	
	\small		
	\linespread{1.45}\selectfont		
	\SetAlgoNoLine%
	\SetKwFor{Foreach}{for each}{do}{end}		
	Initialize: Polyblock $\mathcal{B}^{(1)}$ with a vertex set $\mathcal{L}^{(1)}$, iteration index $m = 1$, vertex $\boldsymbol{y}^{(1)}$, and error tolerance $\epsilon_{\text{PA}}$.	\\	
	\Repeat{$\frac{\norm{\boldsymbol{y}^{(m)} - \boldsymbol{\chi}^{(m)}}}{\norm{\boldsymbol{y}^{(m)}}}\leq \epsilon_{\normalfont{\text{PA}}}$ }{		
		1. Construct a smaller polyblock $\mathcal{B}^{(m+1)}$ with vertex set $\mathcal{L}^{(m+1)}$ by replacing $\boldsymbol{y}^{(m)}$ with $N^\text{T}$ new vertices $\mathcal{L}_{\mathcal{B}}^{(m+1)}=\{\tilde{\boldsymbol{y}}_1^{(m+1)}, \tilde{\boldsymbol{y}}_2^{(m+1)}, \dots, \tilde{\boldsymbol{y}}_{N^\text{T}}^{(m+1)}\}$. The new vertex $\tilde{\boldsymbol{y}}_k^{(m+1)}$, $k \in \{1,2,\dots,N^\text{T}\},$ is given by 
		$$ \tilde{\boldsymbol{y}}_k^{(m+1)} = \boldsymbol{y}^{(m)} - (y_k^{(m)} - x_k^{(m)}) \boldsymbol{u}_k,$$ where $x_k^{(m)}$ is the $k^\text{th}$ element of $\boldsymbol{\chi}^{(m)} = \sqrt{\nu} \frac{\boldsymbol{y}^{(m)}}{\norm{\boldsymbol{y}^{(m)}}}$ and $\boldsymbol{u}_k$ is a unit vector containing only one non-zero element at position $k$. 
		\\
		2. Find $\boldsymbol{y}^{(m+1)} \in \mathcal{L}^{(m+1)}$ as the vertex that maximizes ${\Psi}(\boldsymbol{y})$, i.e., $$\boldsymbol{y}^{(m+1)} = \argmax_{\boldsymbol{y} \in \mathcal{L}^{(m+1)}} \{{\Psi}(\boldsymbol{y})\} $$ \\									
		3. Set $m = m+1$\\ 
	}	
	\textbf{Output:} 
	$\boldsymbol{w}^* = \argmax_{\boldsymbol{y} \in \{\boldsymbol{\chi}^{(1)}, \boldsymbol{\chi}^{(2)},\dots,\boldsymbol{\chi}^{(m-1)} \}} \Psi(\boldsymbol{y})$, $ \Phi(\nu) = \Psi(\boldsymbol{w}^*)$
	\caption{\strut Polyblock outer approximation algorithm}
	\label{OptimalAlgorithm}
\end{algorithm}

We note that the computational complexity of Algorithm~\ref{OptimalAlgorithm} increases exponentially with the number of antennas $N^\text{T}$ employed at the TX.
Therefore, obtaining the optimal value of function $\Phi(\nu)$ may not be feasible in practical multi-user MIMO WPT systems. 
Nevertheless, the obtained optimal solution provides a performance upper-bound for any suboptimal scheme. 
In the next section, we propose an iterative low-complexity algorithm to obtain a suboptimal solution of (\ref{Eqn:MimoPropositionFunction}).

			\subsubsection{Suboptimal Solution}
			\label{Section:SuboptimalSolution}
			In the following, we consider a practical EH model $\phi(\cdot)$ that satisfies Assumption~\ref{Assumption:ConvexityBoundness}.
For this model, we propose an iterative low-complexity algorithm based on \gls*{sdr} and \gls*{sca} to determine a suboptimal solution of (\ref{Eqn:MimoPropositionFunction}).
To this end, we first define matrix $\boldsymbol{W} = \boldsymbol{w} \boldsymbol{w}^H$ and reformulate problem (\ref{Eqn:MimoPropositionFunction}) equivalently as follows:
\begin{subequations}
	\begin{align}
	\maximize_{\boldsymbol{W} \in \mathcal{S}_{+}} \quad &\hat{\Psi}(\boldsymbol{W})
	\label{Eqn:MimoSuboptimalFunctionRef_Obj} \\
	\subjectto \quad &\Tr{\boldsymbol{W}} \leq \nu, \label{Eqn:MimoSuboptimalFunctionRef_C1} \\
	&\rank \{ \boldsymbol{W} \} = 1, \label{Eqn:MimoSuboptimalFunctionRef_C2}
	\end{align}	
	\label{Eqn:MimoSuboptimalFunctionRef}
\end{subequations}
\noindent \hspace*{-5pt}where $\hat{\Psi}(\boldsymbol{W}) = \sum_{m = 1}^{M} \sum_{p = 1}^{N^\text{E}_m} \xi_m \phi(\boldsymbol{g}^m_p \boldsymbol{W} {\boldsymbol{g}^m_p}^H) $ and $\mathcal{S}_{+}$ denotes the set of positive semidefinite matrices.
Since the objective function in (\ref{Eqn:MimoPropositionFunction}) is monotonic non-decreasing in $|\boldsymbol{w}| = \{|\boldsymbol{w}_1| \, |\boldsymbol{w}_2|  \dots  |\boldsymbol{w}_{N^\text{T}}|\}^\top$, we relax the equality constraint in (\ref{Eqn:MimoPropositionFunction}) by inequality constraint (\ref{Eqn:MimoSuboptimalFunctionRef_C1}) \cite{Zhang2012}.

Optimization problem (\ref{Eqn:MimoSuboptimalFunctionRef}) is non-convex due to the non-concavity of objective function (\ref{Eqn:MimoSuboptimalFunctionRef_Obj}) and the non-convexity of constraint (\ref{Eqn:MimoSuboptimalFunctionRef_C2}). 
Therefore, in order to obtain a suboptimal solution of (\ref{Eqn:MimoSuboptimalFunctionRef}), we first eliminate constraint (\ref{Eqn:MimoSuboptimalFunctionRef_C2}).
Then, we denote the total number of rectennas at the EH nodes by $K = \sum_{m=1}^M N^\text{E}_m$ and define the sets $\mathcal{W}_k, k\in\{0,1,\dots\hspace*{-2pt},K\},$ of matrices $\boldsymbol{W}$ such that $\forall \boldsymbol{W} \in \mathcal{W}_k$ exactly $k$ rectifiers are driven into saturation.
We note that $\mathcal{W}_1 \cup \mathcal{W}_2 \cup \dots \cup \mathcal{W}_K = \mathcal{S}_{+}$.
Furthermore, rectifier $p$ of EH node $m$ is driven into saturation if and only if $\boldsymbol{g}^{m}_{p} \boldsymbol{W} {\boldsymbol{g}^{m}_{p}}^H \geq A_s^2$.
Hence, set $\mathcal{W}_k, k\in\{0,1,\dots\hspace*{-2pt},K\},$ consists of $T_k = \frac{K!}{k!(K-k)!}$ convex subsets, where $k!$ denotes the factorial of $k$, i.e., $\mathcal{W}_k = \mathcal{W}_k^1 \cup \mathcal{W}_k^2 \cup \dots \cup \mathcal{W}^{T_k}_k$.
Each convex subset $\mathcal{W}^t_k, t\in\{1,2,\dots\hspace*{-2pt},T_k\},$ consists of all matrices $\boldsymbol{W}$ which drive into saturation a specific combination of $k$ rectennas.
We note that the objective function in (\ref{Eqn:MimoSuboptimalFunctionRef}) is convex for each of these subsets and, hence, applying SCA for solving (\ref{Eqn:MimoSuboptimalFunctionRef}) for $\boldsymbol{W} \in \mathcal{W}^t_k$, is promising \cite{Sun2017, Lanckriet2009}.
Thus, the solution of (\ref{Eqn:MimoSuboptimalFunctionRef}) can be obtained by exploring the subsets $\mathcal{W}^t_k, t\in\{1,2,\dots\hspace*{-2pt},T_k\}, k\in\{0,1,\dots\hspace*{-2pt},K\},$ and solving the resulting problem for each subset \cite{Shanin2021a}.
However, since the computational complexity of this exploration grows with $K!$, in the following, we obtain a suboptimal solution of (\ref{Eqn:MimoSuboptimalFunctionRef}).

For a given transmit power limit $\nu$, we first determine a set of rectennas $\mathcal{W}^*$, which will be driven into saturation.
To this end, we sort the channel gain vectors $\boldsymbol{g}_p^m$ in descending order of their norms as follows $\norm{\boldsymbol{g}_{p_1}^{m_1}} \geq \norm{\boldsymbol{g}_{p_2}^{m_2}} \geq \dots \geq \norm{\boldsymbol{g}_{p_{K}}^{m_{K}}}$, where $m_k \in \{1,2,\dots\hspace*{-2pt},M\}$, $p_k \in \{1,2,\dots\hspace*{-2pt},N_{m_k}^\text{E}\}$, and $k = 1,2,\dots\hspace*{-2pt},K$.
Then, we check if it is possible to drive the $k$ rectifiers with the best channel conditions, i.e., rectifier $p_1$ of EH node $m_1$, rectifier $p_2$ of EH node $m_2$, $\dots$, rectifier $p_k$ of EH node $m_k$, into saturation by solving the following optimization problem:
\begin{subequations}
	\begin{align}
	\maximize_{\boldsymbol{W} \in \mathcal{S}_{+}} \quad & 1 \\
	\subjectto \quad &\boldsymbol{g}^{m_n}_{p_n} \boldsymbol{W} {\boldsymbol{g}^{m_n}_{p_n}}^H \geq A_s^2, \; \forall n \in \{1,2,\dots\hspace*{-2pt},k\}, \label{Eqn:MimoSuboptimalFeasibility_C2a} \\
	 &\boldsymbol{g}^{m_{\tilde{n}}}_{p_{\tilde{n}}} \boldsymbol{W} {\boldsymbol{g}^{m_{\tilde{n}}}_{p_{\tilde{n}}}}^H < A_s^2, \; \forall {\tilde{n}} \in \{k+1,k+2,\dots\hspace*{-2pt},K\}, \label{Eqn:MimoSuboptimalFeasibility_C2b} \\
	& \Tr{\boldsymbol{W}} \leq \nu. \label{Eqn:MimoSuboptimalFeasibility_C2c}
	\end{align}	
	\label{Eqn:MimoSuboptimalFeasibility}
\end{subequations}
\noindent\hspace*{-5pt}Optimization problem (\ref{Eqn:MimoSuboptimalFeasibility}) is convex and can be solved with standard numerical optimization tools, such as CVX \cite{Grant2015}. 
Furthermore, although we dropped the rank-one constraint (\ref{Eqn:MimoSuboptimalFunctionRef_C2}), it can be shown that if (\ref{Eqn:MimoSuboptimalFeasibility}) is feasible and $k > 0$, a beamforming matrix $\boldsymbol{W}^*_k$, which solves (\ref{Eqn:MimoSuboptimalFeasibility}), has rank one. 
A corresponding proof is provided in the conference version \cite[Appendix B]{Shanin2021b} but is omitted here due to the space constraints.
We denote by $k^*$ the maximum number of rectifiers $k$, for which problem (\ref{Eqn:MimoSuboptimalFeasibility}) is feasible.
Note that if (\ref{Eqn:MimoSuboptimalFeasibility}) is not feasible for any $k>0$, we have $k^* = 0$.
Then, we define the convex subset $\mathcal{W}^*$ that corresponds to the case, where the $k^*$ rectifiers with the best channel conditions are driven into saturation.
This set is given by 
\begin{align}
	\mathcal{W}^* = \big\{\boldsymbol{W}: \boldsymbol{W}\in\mathcal{S}_{+}, \boldsymbol{g}^{m_n}_{p_n} \boldsymbol{W} {\boldsymbol{g}^{m_n}_{p_n}}^H &\geq A_s^2, \forall {n} \in \{1,2,\dots\hspace*{-2pt},k^*\}, \nonumber \\ \boldsymbol{g}^{m_{\tilde{n}}}_{p_{\tilde{n}}} \boldsymbol{W} {\boldsymbol{g}^{m_{\tilde{n}}}_{p_{\tilde{n}}}}^H &< A_s^2, \forall {\tilde{n}} \in \{k^*+1,k^*+2,\dots\hspace*{-2pt},K\}\big\}.
	\label{Eqn:MimoFeasibleSet}
\end{align}
Next, we reformulate problem (\ref{Eqn:MimoSuboptimalFunctionRef}) as follows:
\begin{equation}
	\maximize_{\boldsymbol{W} \in \mathcal{W}^*} \quad  \hat{\Psi}(\boldsymbol{W}) \qquad \subjectto \quad \Tr{\boldsymbol{W}} \leq \nu.
	\label{Eqn:MimoSuboptimalConvexRef}
\end{equation}
\noindent Optimization problem (\ref{Eqn:MimoSuboptimalConvexRef}) is still non-convex due to the non-concavity of the objective function.
In the following, we propose to solve (\ref{Eqn:MimoSuboptimalConvexRef}) exploiting \gls*{sca} \cite{Sun2017}.
To this end, we construct an underestimate of the objective function $\hat{\Psi}(\boldsymbol{W})$, which is convex in $\mathcal{W}^*$, as follows:
\begin{equation}
\hat{\Psi}(\boldsymbol{W}) \geq \hat{\Psi}(\boldsymbol{W}^{(t)}) + \Tr{\triangledown\hat{\Psi}(\boldsymbol{W}^{(t)}) (\boldsymbol{W} - \boldsymbol{W}^{(t)} )} ,
\end{equation}
\noindent where $\boldsymbol{W}^{(t)}$ is the solution obtained in the iteration $t$ of the algorithm and $\triangledown \hat{\Psi}(\boldsymbol{W}^{(t)})$ denotes the gradient of $\hat{\Psi}(\boldsymbol{W})$ evaluated at $\boldsymbol{W}^{(t)}$.
Thus, in each iteration $t$ of the proposed algorithm, we solve the following optimization problem:
\begin{equation}
	\boldsymbol{W}^{(t+1)} = \argmax_{\boldsymbol{W}\in \mathcal{W}^*} \; \hat{\Psi}(\boldsymbol{W}^{(t)}) + \Tr{\triangledown\hat{\Psi}(\boldsymbol{W}^{(t)}) (\boldsymbol{W} - \boldsymbol{W}^{(t)} )}\quad
	\subjectto \; \Tr{\boldsymbol{W}} \leq \nu.
	\label{Eqn:MimoSuboptimalFunctionAlg}
\end{equation}
\noindent We note that (\ref{Eqn:MimoSuboptimalFunctionAlg}) is a feasible convex optimization problem that can be solved with standard numerical optimization tools, such as CVX \cite{Grant2015}.
Furthermore, it can be shown that similarly to problem (\ref{Eqn:MimoSuboptimalFeasibility}), the solution of (\ref{Eqn:MimoSuboptimalFunctionAlg}) yields a matrix, whose rank is equal to one.
Hence, we obtain the beamforming vector $\boldsymbol{w}^*$ as the dominant eigenvector of the solution $\boldsymbol{W}^*$ of (\ref{Eqn:MimoSuboptimalFunctionRef}) and compute the corresponding value of function $\Phi(\nu)~=~\Psi(\boldsymbol{w^*})$.
The proposed algorithm is summarized in Algorithm~\ref{OptimizationAlgorithmFunction}.
We note that the proposed algorithm converges to a stationary point of (\ref{Eqn:MimoSuboptimalFunctionRef}) \cite{Lanckriet2009}.
The computational complexity of a single iteration of the algorithm is given by\footnotemark\hspace*{0pt} $\mathcal{O}(K {N^\text{T}}^{\frac{7}{2}} + K^2 {N^\text{T}}^{\frac{5}{2}} +\sqrt{N^\text{T}} K^3)$, where $\mathcal{O}(\cdot)$ is the big-O notation.
\footnotetext{The computational complexity of a convex semidefinite problem that involves an $n \times n$ positive semidefinite matrix and $m$ constraints is given by $\mathcal{O}\big(\sqrt{n} (m n^3 + m^2 n^2 + m^3)\big)$ \cite{Polik2010}. Here, $n = N^\text{T}$ and $m = K+1$.}

\begin{algorithm}[!t]	
	\small				
	\linespread{1.45}\selectfont
	\SetAlgoNoLine%
	\SetKwFor{Foreach}{for each}{do}{end}		
	Initialize: Transmit power $\nu$, tolerance error $\epsilon_\text{SCA}$.	\\	
	1. Sort the channel gain vectors by their norms $\norm{\boldsymbol{g}_{p_1}^{m_1}} \geq \norm{\boldsymbol{g}_{p_2}^{m_2}} \geq \dots \geq \norm{\boldsymbol{g}_{p_{K}}^{m_{K}}}$, where $K = \sum_{m=1}^M N^\text{E}_m$, $m_k \in \{1,2,\dots\hspace*{-2pt},M\}$, $p_k \in \{1,2,\dots\hspace*{-2pt},N_{m_n}^\text{E}\}$, and $k = 1,2,\dots\hspace*{-2pt},K$.\\
	2. Set initial value $k^* = 0$.\\
	\For{$j = 1$ {\upshape to} $K+1$}{		
		3. Solve optimization problem (\ref{Eqn:MimoSuboptimalFeasibility}) for $k = j$ and store $k^* = j$ if the problem is feasible \\						
	}
	4. Determine set $\mathcal{W}^*$, set initial values $h^{(0)} = 0$ and $t = 0$, and randomly initialize $\boldsymbol{W}^{(0)}$. \\
	\Repeat{$|h^{(t)}-h^{(t-1)}|\leq \epsilon_\text{\upshape SCA}$ }{		
		a. For given $\boldsymbol{W}^{(t)}$, obtain $\boldsymbol{W}^{(t+1)}$ as the solution of (\ref{Eqn:MimoSuboptimalFunctionAlg}) \\								
		b. Evaluate $h^{(t+1)} = \hat{\Psi}(\boldsymbol{W}^{(t+1)})$\\
		c. Set $t = t+1$\\ 
	}
	5. Obtain $\boldsymbol{w}^*$ as the dominant eigenvector of $\boldsymbol{W}^{(t)}$ and evaluate $\Phi(\nu) = \Psi(\boldsymbol{w}^*)$\\
	\textbf{Output:} 
	$ \Phi(\nu)$, $\boldsymbol{w}^*$
	\caption{\strut Suboptimal algorithm for solving optimization problem (\ref{Eqn:MimoPropositionFunction}) }
	\label{OptimizationAlgorithmFunction}
\end{algorithm}	
				
	\section{Numerical Results}
	In this section, we evaluate the performance of the proposed transmit strategies via simulations.
First, we compare the performance of single-user MISO, SIMO, and MIMO WPT systems.
Then, we evaluate the performance of multi-user MIMO WPT systems for different numbers of antennas at the TX and EH nodes.
Finally, we study the influence of the weights $\xi_m, m\in \{1,2,\dots,M\},$ associated with the EH nodes and determine the harvested power region.

		\subsection{Simulation Parameters}
		
In the following, we provide the system parameters adopted in our simulations.
In our simulations, the path losses are calculated as $35.3+37.6\log_{10}(d_m)$, where $d_m$ is the distance between the TX and EH node $m$ \cite{Ghanem2020}.
Furthermore, in order to harvest a meaningful amount of power, we assume that the TX and each EH node have a line-of-sight link.
Thus, the channel gains $\boldsymbol{g}_p^m$ follow Rician distributions with Rician factor $1$ \cite{Goldsmith2005}.
For the EH model $\phi(\cdot)$, we adopt the model proposed in \cite{Morsi2019} and given by (\ref{Eqn:RaniaModel}) with parameter values {$a=1.29$, $B = 1.55\cdot10^3$, $I_s = \SI{5}{\micro\ampere}$, $R_L = \SI{10}{\kilo\ohm}$, and $A_s^2 = \SI{25}{\micro\watt} $}.
For Algorithms~\ref{AlgorithmGridSearch}, \ref{OptimalAlgorithm}, and \ref{OptimizationAlgorithmFunction} we adopt step size, grid size, and error tolerance values of $\Delta_{\rho} = 0.1$, $N_{\rho} = 10^3$, $\epsilon_{\text{PA}} = 10^{-3}$, and $\epsilon_{\text{SCA}} = 10^{-3}$, respectively.
We note that the grid size $N_{\rho}$ is chosen sufficiently large to ensure that the function $\Phi(\nu)$ saturates for $\nu = \rho_{N_\rho}$.
All simulation results were averaged over $1000$ channel realizations.
		\subsection{Single-user WPT Systems}
		\begin{figure}[!t]
	\centering
	\includegraphics[width=0.45\textwidth, draft=false]{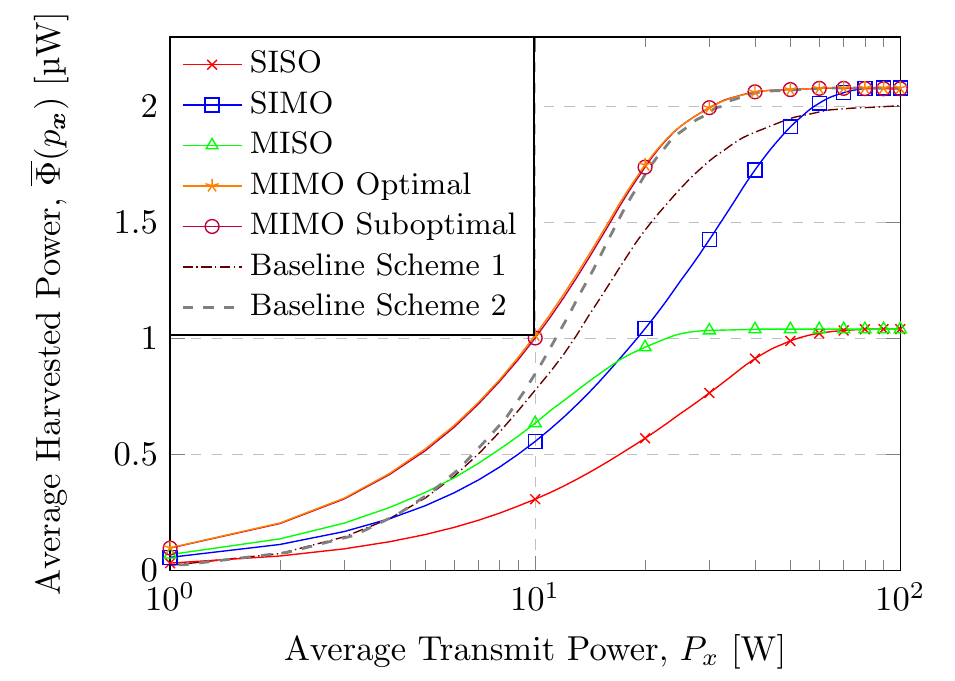}
	\caption{Comparison of single-user SISO, SIMO, MISO, and MIMO WPT systems.}
	\label{Fig:Results_SU}
\end{figure}
In this section, we investigate the performance of single-user WPT systems with different numbers of antennas at the TX and the EH node.
The distance between the TX and the EH is $d = \SI{10}{\meter}$.
The considered MISO and SIMO WPT systems employ $N^\text{T} = 2$ and $N^\text{E} = 2$ antennas at the TX and the EH node, respectively. 
For these systems, we evaluate the average harvested power $\overline{\Phi}(p^*_{\boldsymbol{x}})$ for different values of the power budget $P_x$ by applying Corollaries~\ref{Theorem:Corollary3} and \ref{Theorem:Corollary4}, respectively.
In the considered single-user MIMO WPT system, the TX and EH nodes employ $N^\text{T} = N_1^\text{E} = 2$ antennas, respectively.
For this system, optimal and suboptimal pdfs $p_{\boldsymbol{x}}(\boldsymbol{x})$ are determined by combining Algorithm~\ref{AlgorithmGridSearch} with Algorithms~\ref{OptimalAlgorithm} and \ref{OptimizationAlgorithmFunction}, respectively, and the values of $\overline{\Phi}(p^*_{\boldsymbol{x}})$ are computed.
For comparison, we also consider a SISO WPT system employing the optimal transmit strategy in \cite{Morsi2019}.
As Baseline Scheme 1, we consider a MIMO WPT system with energy beamforming at the TX, which is optimal for linear EHs \cite{Zhang2013}.
Furthermore, as Baseline Scheme 2, similar to \cite{Shen2020}, we consider a MIMO WPT system, where a scalar input symbol and a single beamforming vector are adopted at the TX. 
For this system, we obtain the optimal beamforming vector $\boldsymbol{w}^*$ as solution of (\ref{Eqn:MimoPropositionBeamformerProblem}) for $\nu = P_x$ with Algorithm~\ref{OptimalAlgorithm} and compute the corresponding harvested power as $\overline{\Phi} = \Phi(\boldsymbol{w}^*)$.

In Fig.~\ref{Fig:Results_SU}, we plot the average harvested powers $\overline{\Phi}(p^*_{\boldsymbol{x}})$ for different values of $P_x$.
First, we observe that for each considered WPT setup, the average harvested power $\overline{\Phi}(\cdot)$ is bounded above, since for the EH model in (\ref{Eqn:RaniaModel}), for sufficiently large values of $P_x$, all rectifiers of the EH node are driven into saturation.
Furthermore, the saturation level of the harvested power is proportional to the number of rectennas employed at the EH.
As expected, the MIMO WPT system achieves a superior performance compared to the SIMO and MISO WPT systems, which, in turn, outperform the SISO WPT system  significantly.
Interestingly, the MISO WPT system can harvest more power than the SIMO WPT system if the transmit power budget at the TX is low, whereas, for large values of $P_x$, the single rectenna of the MISO WPT system is driven into saturation and, thus, more power can be harvested by the SIMO WPT system.
Furthermore, we observe that for the MIMO WPT system, the proposed optimal transmit strategy, which employs two beamforming vectors, outperforms Baseline Schemes 1 and 2, which employ a single beamforming vector.
However, we note that Baseline Scheme 2, where the optimal beamforming vector obtained with Algorithm~\ref{OptimalAlgorithm} is utilized, outperforms Baseline Scheme 1, where energy beamforming is adopted \cite{Zhang2013}.
Finally, we note that although the computational complexity of the proposed suboptimal scheme for obtaining the MIMO WPT beamforming vectors is significantly lower than that of the optimal scheme, the performances attained by both schemes are practically identical. 
Therefore, in the next section, in order to keep the computational complexity low, we adopt the suboptimal scheme to evaluate the performance of multi-user MIMO WPT systems.

		\subsection{Multi-user WPT Systems}
		\begin{figure}[!t]
	\centering
	\subfigure[Comparison for different numbers of EH nodes, $M$]{
		\includegraphics[width=0.45\textwidth, draft=false]{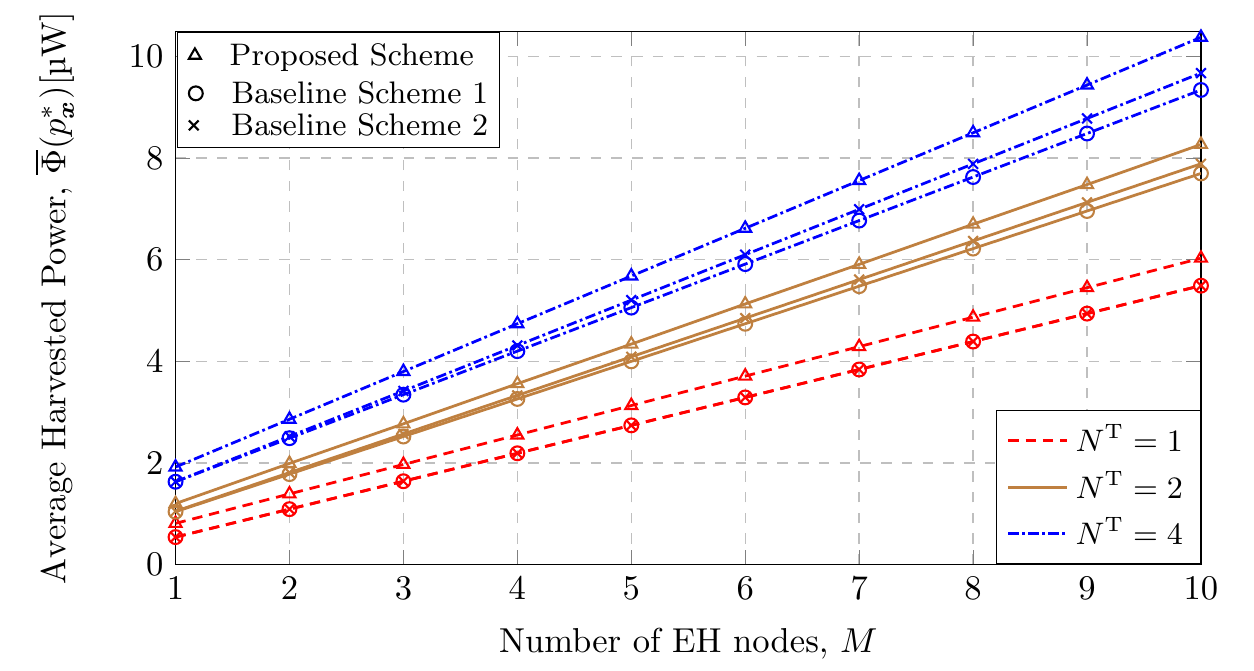} \label{Fig:Results_MU_Neh}}
	\quad
	\subfigure[Comparison for different numbers of TX antennas, $N^{\text{T}}$]{
		\includegraphics[width=0.45\textwidth, draft=false]{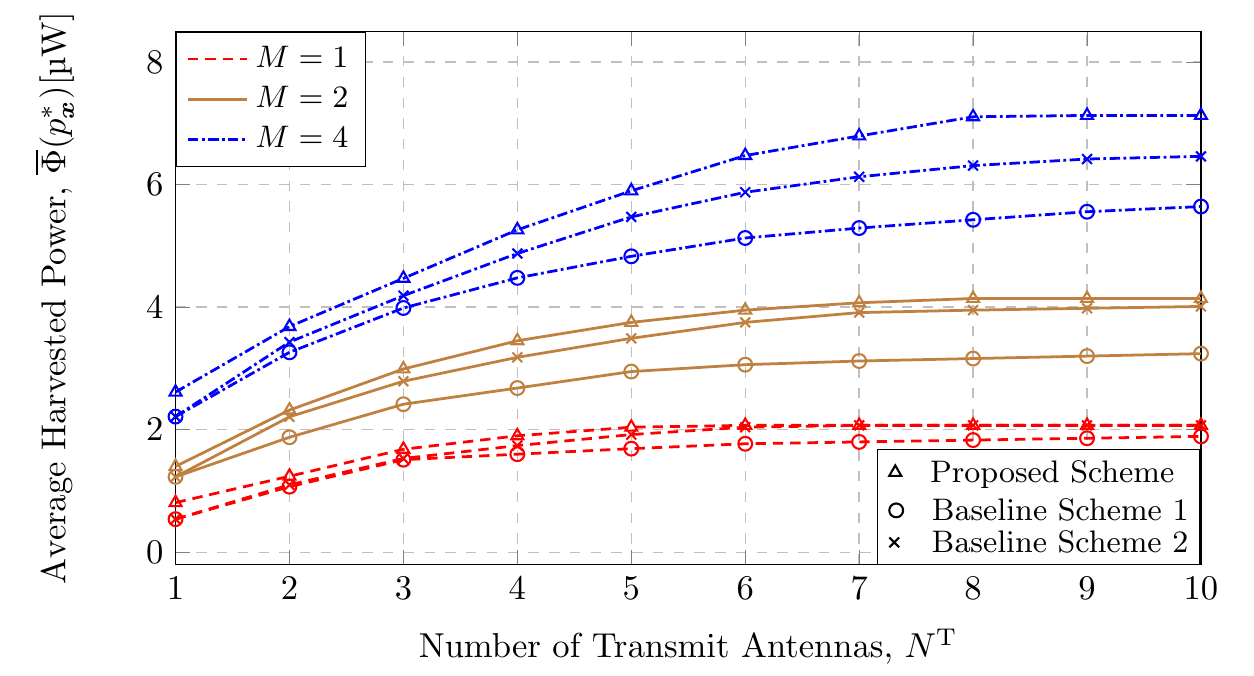} \label{Fig:Results_MU_Nt}}
	\caption{Average harvested power ${\overline{\Phi}(p^*_{\boldsymbol{x}})}$ for different numbers of transmit antennas $N^\text{T}$ and EH nodes $M$.}
	\label{Fig:Results_MU}
\end{figure}
In this section, we consider multi-user MIMO WPT systems, where the EH nodes are equipped with $N^{\text{E}}_m = 2, m \in \{1,2,\dots,M\},$ rectennas.
In Fig. \ref{Fig:Results_MU_Neh} and \ref{Fig:Results_MU_Nt}, we show the average harvested power for different numbers of transmit antennas $N^\text{T}$ and EH nodes $M$, respectively.
For each system setup, we compare the performance of the proposed transmit strategy with Baseline Scheme 1 and Baseline Scheme 2.
The results in Fig.~\ref{Fig:Results_MU} are obtained assuming a transmit power budget of $P_x = \SI{15}{\watt}$, equal weights for all EH nodes, i.e., $\xi_m = \frac{1}{M}$, and equal distances of $d_m = \SI{10}{\meter}, m = \{1,2,\dots,M\},$ between the TX and EH nodes.
We observe that higher values of $N^\text{T}$ and $M$ yield larger average harvested powers $\overline{\Phi}$.
Furthermore, we note that, similar to the single-user case, the proposed transmit strategy yields a better performance than the baseline schemes.
Moreover, for SIMO WPT systems, i.e., for $N^\text{T} = 1$, the transmit strategies for Baseline Scheme 1 and Baseline Scheme 2 are identical and depend only on the power budget $P_x$.
In Fig.~\ref{Fig:Results_MU_Neh}, we observe that the harvested power depends practically linearly on the number of EH nodes.
Although this result is relatively straightforward if the EH nodes are driven into saturation, see Fig.~\ref{Fig:Results_SU}, Fig.~\ref{Fig:Results_MU_Neh} suggests that the linear growth of $\overline{\Phi}(\cdot)$ as $M$ increases also holds when the rectifiers are not saturated. 
On the contrary, in Fig.~\ref{Fig:Results_MU_Nt}, we observe that for a large number of transmit antennas, the average harvested power saturates since for the EH model in (\ref{Eqn:RaniaModel}), the harvested power is bounded above. 
In fact, a larger number of transmit antennas $N^\text{T}$ enables a more efficient exploitation of the transmit power budget, which yields a higher received power at the EH nodes.
Furthermore, for larger $M$, the harvested power is distributed among a larger number of EH nodes and, thus, the number of transmit antennas needed to drive the rectennas into saturation, grows with $M$.
We observe that for small numbers of the rectennas at the EH nodes, i.e., for $M = 1$ and $M = 2$, the saturation level of the average harvested power is nearly identical for the proposed scheme and Baseline Scheme 2, whereas, for Baseline Scheme~1, this saturation level is significantly lower.
In fact, for massive MIMO systems, where the number of transmit antennas is much larger than the number of rectennas, the channel vectors $\boldsymbol{g}_p^m$ are practically orthogonal and, hence, the transmit strategy of Baseline Scheme 1 favors the rectenna with the best channel conditions while the powers transferred to the other rectennas are significantly lower \cite{Zhang2013}, which may not be optimum in case of rectenna saturation.

\begin{figure}[!t]
	\centering
	\subfigure[Low transmit power regime, $P_x = \SI{10}{\watt}$]{
		\includegraphics[width=0.45\textwidth, draft=false]{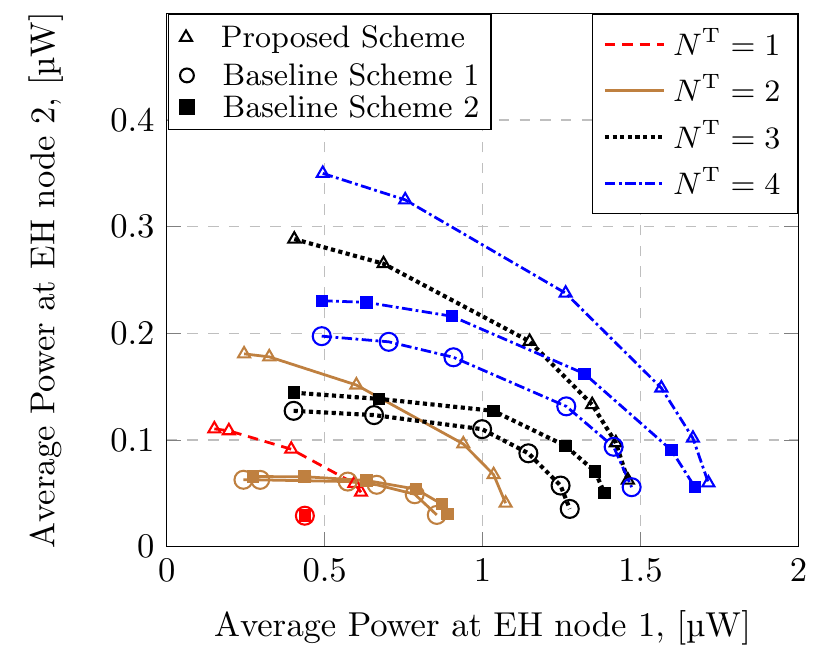} \label{Fig:Results_MU_Region_LP}}
	\quad
	\subfigure[High transmit power regime, $P_x = \SI{30}{\watt}$]{
		\includegraphics[width=0.45\textwidth, draft=false]{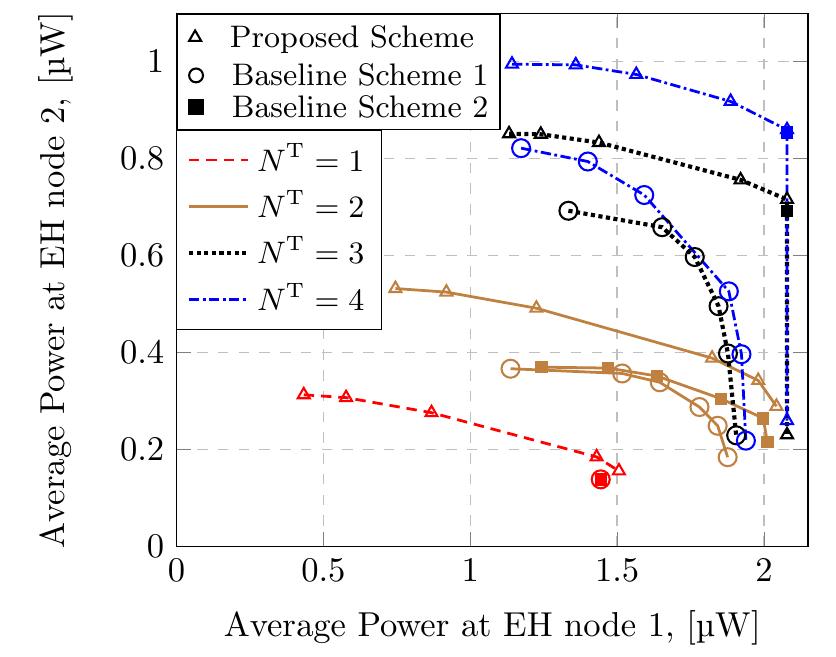} \label{Fig:Results_MU_Region_HP}}
\caption{Harvested power regions of the MIMO WPT system with $M=2$ EH nodes equipped with $N^\text{E}_m = 2, m=\{1,2\},$ rectennas for different numbers of transmit antennas $N^\text{T}$.}
	\label{Fig:Results_MU_Region}
\end{figure}
Finally, in Fig.~\ref{Fig:Results_MU_Region}, we study the influence of weights $\xi_m, m\in\{1,2\},$ for $M=2$ EH nodes on the powers harvested by the individual EH nodes.
For the results in Fig.~\ref{Fig:Results_MU_Region}, by varying the weights $\xi_m$, we determine the optimal transmit strategies that maximize the corresponding weighted sum of the average powers harvested at the EH nodes.
Then, for the optimal pdfs $p_{\boldsymbol{x}}^*(\boldsymbol{x})$, we evaluate and plot the average harvested powers at the individual EH nodes, $\mathbb{E}_{\boldsymbol{x}} \{\psi_m(\boldsymbol{x})\}, m=\{1,2\},$ respectively.
For comparison, we also show the average harvested powers obtained with Baseline Scheme 1 and Baseline Scheme~2, respectively.
In Fig.~\ref{Fig:Results_MU_Region_LP}, we consider a low transmit power regime characterized by a power budget of $P_x = \SI{10}{\watt}$, whereas for the results in Fig.~\ref{Fig:Results_MU_Region_HP}, we assume a high transmit power regime with $P_x = \SI{30}{\watt}$.
The distances between the EH nodes and the TX are equal to $d_1 = \SI{10}{\meter}$ and $d_2 = \SI{25}{\meter}$, respectively.
In Fig.~\ref{Fig:Results_MU_Region}, we observe that higher values of $N^\text{T}$ and $P_x$ yield larger average harvested powers at both EH nodes.
Furthermore, as expected, the proposed scheme yields a better performance compared to Baseline Scheme 1 and Baseline Scheme 2.
For both transmit power regimes, we observe that for SIMO WPT, i.e., $N^\text{T} = 1$, the performances obtained with both baseline schemes are identical and do not depend on the adopted weights $\xi_1$ and $\xi_2$.
In fact, in this case, the transmit strategies of the baseline schemes are identical and depend on the power budget $P_x$ only.
Moreover, in the high transmit power regime, we observe that the performance of Baseline Scheme 2 also does not depend on $\xi_m, m = \{1,2\},$ for large numbers of transmit antennas, i.e., $N^\text{T} = \{3,4\}$, since, for large $N^\text{T}$, EH node $1$ is driven into saturation anyways.
However, for the other system setups, the choice of the weights $\xi_m, m\in\{1,2\},$ enables a trade-off between the powers harvested at the EH nodes, which is characterized by a convex \emph{harvested power region}. 
Furthermore, by increasing weight $\xi_m$, more power is harvested at EH node $m$ at the expense of a reduction of the power harvested by the other node. 
Thus, by choosing the user weights, the TX can control the distribution of the harvested power among the users. 
In particular, for $\xi_1 = 1$ (and $\xi_2 = 0$), the TX maximizes the average harvested power at EH node $1$ and neglects EH node $2$, which may yield a substantial decrease of the power at EH node $2$.
In the high transmit power regime, EH node $1$ is driven into saturation for $N^\text{T} = \{3,4\}$.
In this case, by decreasing $\xi_1$ (and increasing $\xi_2$), it is possible to significantly increase the power harvested by EH node $2$ without a substantial reduction of the power harvested by EH node $1$.

	\section{Conclusion}
	In this paper, we considered multi-user MIMO WPT systems with multiple EH nodes employing non-linear rectennas. 
Based on a set of assumptions, which are satisfied for practical EH circuits, we specified a general EH model.
Then, we proposed an optimal transmit strategy that maximizes the weighted sum of the average harvested powers at the EH nodes under a constraint on the power budget of the TX.
For MISO WPT, we showed that transmission of scalar symbols with discrete random magnitudes, whose pdf has at most two mass points, via MRT beamforming is optimal.
Next, for SIMO WPT, we proved that the optimal transmit symbol magnitude also has a discrete pdf with no more than two mass points.
Then, for MIMO WPT, we showed that the optimal transmit strategy employs a scalar unit norm symbol and at most two beamforming vectors.
In order to obtain these vectors, we proposed and solved a non-convex optimization problem.
Since the computational complexity of the optimal solution was high, we developed an iterative low-complexity algorithm to obtain a suboptimal solution.
Our simulation results revealed that the proposed optimal and suboptimal schemes for MIMO WPT systems yield practically identical performance.
Furthermore, we observed that the proposed MIMO WPT design achieves substantial performance gains compared to two baseline schemes, one based on a linear EH model and the other one based on a single beamforming vector.
Moreover, for multi-user MIMO WPT systems, we showed that the harvested power saturates for large numbers of TX antennas.
Finally, we observed a trade-off between the powers harvested at individual EH nodes, which we characterized by a convex harvested power region.
			
	\appendices
	\begin{appendices}
		\renewcommand{\thesection}{\Alph{section}}
		\renewcommand{\thesubsection}{\thesection.\arabic{subsection}}
		\renewcommand{\thesectiondis}[2]{\Alph{section}:}
		\renewcommand{\thesubsectiondis}{\thesection.\arabic{subsection}:}	
		\section{Proof of Lemma 1}
		\label{Appendix:LemmaProof}
		In the following, we prove Lemma \ref{Theorem:Lemma1}.
First, we note that since $\nu_2^*$ is the maximizer of the slope function $\sigma(\cdot, \cdot; f)$ for $\nu_1 = {\nu}^*_1$, then we have
\begin{equation}
\frac{f(\nu_2^*) - f({\nu}^*_1)}{\nu_2^* - {\nu}^*_1} \geq \frac{f({\nu}) - f({\nu}^*_1)}{{\nu} - {\nu}^*_1}, \; \forall \nu \geq \overline{\nu} \geq {\nu}^*_1.
\label{Eqn:LemmaProof1}
\end{equation}
Then, since $\nu_1^*$ is the minimizer of $\sigma(\cdot, \cdot; f)$ for $\nu_2 = \nu_2^*$, we have
\begin{align}
\frac{f(\nu_2^*) - f(\nu_1^*)}{\nu_2^* - \nu_1^*} &\leq \frac{f(\nu_2^*) - f(\nu)}{\nu_2^* - \nu}, \forall \nu \leq \overline{\nu} \;  \Longleftrightarrow \\
&f(\nu_1^*)\nu_2^* - f(\nu)(\nu_2^*-\nu_1^*) \geq f(\nu_2^*)(\nu_1^* - \nu) + f(\nu_1^*)\nu, \forall \nu \leq \overline{\nu}.
\label{Eqn:LemmaProof2}
\end{align} 
Next, we subtract $f(\nu_1^*)\nu_1^*$ from both sides of (\ref{Eqn:LemmaProof2}). 
This allows us to rewrite both (\ref{Eqn:LemmaProof1}) and (\ref{Eqn:LemmaProof2}) as follows:
\begin{equation}
f(\nu_1^*) - f(\nu) \geq \frac{f(\nu_2^*) - f(\nu_1^*)}{\nu_2^* - \nu_1^*} (\nu_1^* - \nu),
\label{Eqn:LemmaProof3}
\end{equation}
\noindent respectively, which, thus, holds $\forall \nu \in \mathbb{R}$.
Let us define linear function $g(\nu) = f(\nu_1^*) + \frac{f(\nu_2^*) - f(\nu_1^*)}{\nu_2^* - \nu_1^*} (\nu - \nu_1^*)$.
Since, from (\ref{Eqn:LemmaProof3}), $g(\nu) \geq f(\nu)$, $\forall \nu \in \mathbb{R}$, then $\mathbb{E}_{\nu} \{f(\nu)\} \leq \mathbb{E}_{\nu} \{g(\nu)\} = g\big(\overline{\nu}\big) = f(\nu_1^*) + \alpha \big(f(\nu_2^*) - f(\nu_1^*) \big)$, where $\alpha = \frac{\overline{\nu} - \nu_1^*}{\nu_2^* - \nu_1^*}$.
Finally, with $\alpha = 1-\beta$, we have $\mathbb{E}_\nu \{f(\nu)\} \leq \beta f(\nu_1^*) + (1-\beta) f(\nu_2^*)$, where the inequality holds with equality if the pdf of $\nu$ is given by $p_\nu^*(\nu) = \beta \delta(\nu - \nu_1^*) + (1-\beta) \delta(\nu - \nu_2^*)$.
This concludes the proof.
		\section{Proof of Corollary 1}
		\label{Appendix:Corollary1Proof}
		To prove this corollary, let us define a linear function $g(\nu) = f'(\overline{\nu})(\nu-\overline{\nu}) + f(\overline{\nu})$.
We note that due to (\ref{Eqn:Corollary1_Condition}), $g(\nu) \geq f(\nu)$ $\forall \nu \in \mathbb{R}$.
Thus, $\mathbb{E}_\nu\{f(\nu)\} \leq \mathbb{E}_\nu\{g(\nu)\} = g\big( \overline{\nu} \big)$, where the equality holds due to the linearity of $g(\nu)$.
Hence, we conclude that the expectation of $f(\cdot)$ is upper-bounded by $\mathbb{E}_\nu\{f(\nu)\} \leq f(\overline{\nu})$, where the inequality holds with equality for pdf $p_\nu(\nu) = \delta(\nu-\overline{\nu})$.
This concludes the proof.
		\section{Proof of Corollary 2}
		\label{Appendix:Corollary2Proof}
		First, we note that the objective function in optimization problem (\ref{Eqn:GeneralOptimizationProblem}) is monotonically non-decreasing, whereas the feasible set of the problem is defined by inequality $\mathbb{E}_\nu\{\nu\} \leq A_\nu$. 
Hence, (\ref{Eqn:GeneralOptimizationProblem}) can be equivalently rewritten as follows \cite{Zhang2012}:
\begin{equation}
\maximize_{ {p}_{\nu} } \; \mathbb{E}_\nu \{f(\nu)\} \quad \subjectto \; \mathbb{E}_\nu \{\nu\} = A_\nu.
\label{Eqn:GeneralOptimizationProblemRef}
\end{equation}
Therefore, if condition (\ref{Eqn:Corollary1_Condition}) holds for function $f(\nu)$ and $\overline{\nu} = A_\nu$, the objective function in (\ref{Eqn:GeneralOptimizationProblemRef}) is upper-bounded by $\mathbb{E}_\nu \{f(\nu)\} \leq f(A_\nu)$, see Corollary~\ref{Theorem:Corollary1}, where the inequality holds with equality if $p_\nu(\nu) = \delta(\nu - A_\nu)$.

Let us now consider the case, where condition (\ref{Eqn:Corollary1_Condition}) does not hold for $f(\nu)$ and $\overline{\nu} = A_\nu$.
In this case, the expectation of $f(\nu)$ is upper-bounded by (\ref{Eqn:LemmaUpperBound}), see Lemma~\ref{Theorem:Lemma1}, where the inequality holds with equality if the pdf of random variable $\nu$ is given by $p_\nu^*(\nu) = \beta \delta(\nu - \nu_1^*) + (1-\beta) \delta(\nu - \nu_2^*)$ with $\nu_1^*$, $\nu_2^*$, and $\beta$ defined as in Lemma~\ref{Theorem:Lemma1}.
This concludes the proof.

		\section{Proof of Proposition 1}
		\label{Appendix:Prop1Proof}
		We solve optimization problem (\ref{Eqn:WPT_GeneralProblem}) for a single-user MISO WPT system, i.e., $M = N^\text{E}_M = 1$.
First, let us consider a distribution of the transmit symbols $\boldsymbol{x}$ which has a point of increase at $\tilde{\boldsymbol{x}}_0$.
For this distribution, a larger value of the input power at the EH and, thus, an equal or larger value of $\overline{\Phi}(\cdot)$ can be attained by removing the mass point $\tilde{\boldsymbol{x}}_0$ and increasing the probability of symbol $\boldsymbol{x}_0 = \norm{\tilde{\boldsymbol{x}}_0} \frac{\boldsymbol{g}^H}{\norm{\boldsymbol{g}}} \exp(j\theta_s)$ by the probability of symbol $\tilde{\boldsymbol{x}}_0$ of the former distribution, see Assumptions \ref{Assumption:IncreasingBounded}\cite{Zhang2013}.
We note that this transformation preserves the validity of the distribution, i.e., $\int p_{\boldsymbol{x}}(\boldsymbol{x}) d\boldsymbol{x} = 1$, and, since the transmit powers for the two symbols are identical, i.e., $\norm{\tilde{\boldsymbol{x}}_0}^2 =  \norm{{\boldsymbol{x}}_0}^2$, the new distribution does not affect the power budget of the TX.

Therefore, for the solution of (\ref{Eqn:WPT_GeneralProblem}), transmit vector  $\boldsymbol{x} = \boldsymbol{w} r_s \exp(j\theta_s)$ is optimal, where $\boldsymbol{w} = \frac{\boldsymbol{g}^H}{ \norm{\boldsymbol{g}}}$ is the MRT beamformer and $r_s = |s|$ and $\theta_s$ are the magnitude and the arbitrary phase of random scalar symbol $s$, respectively. 
We denote the pdf of the transmit power values $\nu = r_s^2$, $\nu \in [0, +\infty)$, by $p_\nu(\nu)$.
Then, the utility function in (\ref{Eqn:EhUtilityFunction}) can be rewritten as a function of pdf $p_\nu(\nu)$ as follows:
\begin{equation}
	\overline{\Phi} = \int_{\boldsymbol{x}} p_{\boldsymbol{x}}(\boldsymbol{x}) \phi(|\boldsymbol{g} \boldsymbol{x}|^2) d\boldsymbol{x} = \int_{r_s} p_{r_s}(r_s) \phi(\norm{\boldsymbol{g}}^2 r_s^2) d r_s = \int_{\nu} p_{\nu}(\nu){\Phi} (\nu) d \nu = \mathbb{E}_{\nu} \{{\Phi} (\nu)\}.
\end{equation}
Hence, problem (\ref{Eqn:WPT_GeneralProblem}) can be equivalently rewritten as follows:
\begin{equation}
\maximize_{{p}_{\nu}} \quad \mathbb{E}_{\nu} \{{\Phi} (\nu)\} \quad
\text{subject to}\quad  \mathbb{E}_{\nu} \{\nu\} \leq P_x.
\label{Eqn:Prop1ProofProblem}
\end{equation}

Since optimization problem (\ref{Eqn:Prop1ProofProblem}) is in the form of auxiliary problem (\ref{Eqn:GeneralOptimizationProblem}), we obtain the solution by applying Corollary~\ref{Theorem:Corollary2}.
First, the optimal pdf is given by $p_\nu^*(\nu) = \delta(\nu - P_x)$ if 
\begin{equation}
\Phi' (P_x) (P_x - \nu) \leq \Phi (P_x) - \Phi (\nu), \forall \nu \in \mathbb{R}_{+}.
\label{Eqn:Prop1ProofCondition}
\end{equation}
We note that condition (\ref{Eqn:Prop1ProofCondition}) is equivalent to (\ref{Eqn:Proposition1Condition}).
Furthermore, since $F_{r_s}(r_s) = F_\nu(r_s^2)$, where $F_{r_s}(r_s)$ and $F_\nu(\nu)$ are the cumulative density functions of $r_s$ and $\nu$, respectively, the optimal pdf of $r_s$ is given by $p_{r_s}^*(r_s) = \delta(r_s - \sqrt{P_x})$.

If (\ref{Eqn:Prop1ProofCondition}) does not hold, according to Corollary~\ref{Theorem:Corollary2}, the optimal solution of (\ref{Eqn:Prop1ProofProblem}) is given by $p^*_{\nu}(\nu) = (1-\beta) \delta(\nu - {\nu^*_1}) + \beta \delta(\nu - {\nu^*_2})$, where $\nu_1^*$ and $\nu_2^*$ are given by (\ref{Eqn:SolutionNu1}) and (\ref{Eqn:SolutionNu2}), respectively.
Finally, since $F_{r_s}(r_s) = F_\nu(r_s^2)$, the equivalent optimal pdf of $r_s$ is given by $p_{r_s}^* = (1-\beta) \delta(r_s - \sqrt{\nu^*_1}) + \beta \delta(r_s - \sqrt{\nu^*_2})$.
This concludes the proof.

		\section{Proof of Proposition 2}
		\label{Appendix:Prop2Proof}
		In order to prove Proposition~\ref{Theorem:Proposition2}, we note that as a sum of non-decreasing functions, function $\Phi(\cdot)$ is also monotonically non-decreasing.
Furthermore, the objective function can be equivalently rewritten as follows:
\begin{equation}
\overline{\Phi}(p_{r_x}) = \int_{r_x} p_{r_x}(r_x) \Phi(r_x^2) dr_x = \int_{\nu} p_{\nu}(\nu){\Phi} (\nu) d \nu = \mathbb{E}_{\nu} \{{\Phi} (\nu)\},
\end{equation}
where $p_\nu(\nu)$ is the pdf of the transmit power $\nu = r_x^2$.

Hence, for the considered SIMO WPT system, optimization problem (\ref{Eqn:WPT_GeneralProblem}) can be equivalently reformulated as follows:
\begin{equation}
	\maximize_{{p}_{\nu}} \quad \mathbb{E}_{\nu} \{\Phi(\nu)\} \quad
	\text{subject to}\quad  \mathbb{E}_{\nu} \{\nu\} \leq P_x.
	\label{Eqn:Prop2ProofProblem}
\end{equation}

Since optimization problem (\ref{Eqn:Prop2ProofProblem}) is in the form of auxiliary problem (\ref{Eqn:GeneralOptimizationProblem}), the application of Corollary~\ref{Theorem:Corollary2} yields Proposition~\ref{Theorem:Proposition2}.
This concludes the proof.
		\section{Proof of Corollary 4}
		\label{Appendix:Corollary4Proof}
		For the considered SIMO WPT systems with two rectennas, $\Phi(r_x^2) = \sum_{p=1}^{2} \phi(r_x^2 |g_p|^2)$.
Hence, the function $\Phi(\nu)$ is monotonic non-decreasing and convex in the intervals $[0, {\rho_\text{min}})$, $[{\rho_\text{min}}, {\rho_\text{max}})$, and $[{\rho_\text{max}}, \infty)$, respectively, and bounded $\Phi(\nu) \in [0, \Phi_r^\text{max}]$, where $\Phi_r^\text{max} = \Phi({\rho_\text{max}}) = 2\phi(A_s^2)$.
Therefore, if affordable by the power budget constraint, i.e., $P_x \geq \rho_\text{max}$, the optimal pdf is given by $p_r^*(r) = \delta(r - \sqrt{\rho_\text{max}})$.

From the condition in Assumption~\ref{Assumption:Quadratic}, we obtain the following inequality:
\begin{equation}
\begin{aligned}
	\frac{\Phi({\rho_\text{min}}) }{\rho_\text{min}} = \frac{\phi(A_s^2) + \phi(A_s^2 \rho_\text{min} / \rho_\text{max})}{\nu_1} \geq \frac{\phi(A_s^2) + \phi(A_s^2) (\rho_\text{min} / \rho_\text{max})^2}{\rho_\text{min}} \\ 
	= \frac{\phi(A_s^2)}{\rho_\text{max}} \frac{\rho_\text{max}^2 + \nu_1^2}{\rho_\text{min} \rho_\text{max}} \geq \frac{2\phi(A_s^2)}{\rho_\text{max}} = \frac{\Phi({\rho_\text{max}})}{\rho_\text{max}} > 0.
	\label{Eqn:Corollary4ProofCondition}
\end{aligned}	
\end{equation}	
Since function $\Phi(\cdot)$ is convex in the intervals $[0, {\rho_\text{min}})$, $[{\rho_\text{min}}, {\rho_\text{max}})$ and $\frac{\Phi({\rho_\text{min}})}{\rho_\text{min}} \geq \frac{\Phi({\rho_\text{max}})}{\rho_\text{max}}$, the solution of the optimization problem in Proposition~\ref{Theorem:Proposition2} is given by $\nu_1^* = 0$ and $\nu_2^* = \rho_\text{min}$ if $P_x < \rho_\text{min}$.
Finally, we note that from (\ref{Eqn:Corollary4ProofCondition}), we have $\frac{\Phi({\rho_\text{max}})}{\rho_\text{max}} \geq \frac{\Phi({\rho_\text{max}}) - \Phi({\rho_\text{min}})}{\rho_\text{max} - \rho_\text{min}}$ and, hence, if $P_x \in [\rho_\text{min}, \rho_\text{max})$, $\nu_1^* = \rho_\text{min}$ and $\nu_2^* = \rho_\text{max}$.
This concludes the proof.
		\section{Proof of Proposition 3}
		\label{Appendix:Prop3Proof}
			First, we note that for any arbitrary transmit symbol $\tilde{\boldsymbol{x}}$, there is a symbol $\hat{\boldsymbol{x}}$ given by
	\begin{equation}
		\hat{\boldsymbol{x}} = \argmax_{\boldsymbol{x}} \Psi(\boldsymbol{x}) \; \subjectto \norm{\boldsymbol{x}}^2 = \norm{\tilde{\boldsymbol{x}}}^2,
	\end{equation}  
	which has the same transmit power and yields a higher or equal value of $\Psi(\boldsymbol{x})$.
	Hence, for any arbitrary distribution of transmit symbols with a point of increase $\tilde{\boldsymbol{x}}$, a larger value of $\Psi(\boldsymbol{x})$ can be obtained by removing this point and increasing the probability of $\hat{\boldsymbol{x}}$ by the corresponding value.
	
	Let us introduce now a function $\Phi(\nu)$ that returns the largest possible value of $\Psi(\boldsymbol{x})$ if a symbol with power $\nu$ was transmitted. 
	This function is given by (\ref{Eqn:MimoPropositionFunction}). 
	We note that function $\Phi(\cdot)$ is monotonically non-decreasing, see Assumption~\ref{Assumption:IncreasingBounded}.
	Then, the solution of (\ref{Eqn:WPT_GeneralProblem}) can be obtained by determining first the solution $p^*_\nu(\nu)$ of the following optimization problem:
	\begin{equation}
		\maximize_{p_{\nu} (\nu)} \; \mathbb{E}_\nu\{\Phi(\nu)\}\; \quad \subjectto \; \mathbb{E}_\nu\{\nu\} \leq P_x.
		\label{Eqn:MimoPropositionProofProblem}
	\end{equation}
	Since (\ref{Eqn:MimoPropositionProofProblem}) is in the form of (\ref{Eqn:GeneralOptimizationProblem}), there exists an optimal discrete pdf $p^*_\nu(\nu)$ consisting of at most two mass points, $\nu^*_1$ and $\nu^*_2$, see Corollary~\ref{Theorem:Corollary2}.
	Hence, the optimal symbol vector $\boldsymbol{x}$ can be decomposed as $\boldsymbol{x} = \boldsymbol{w} s$ with unit-norm symbols $s$ and discrete random beamforming vector $\boldsymbol{w}$, whose pdf consists of at most two mass points evaluated as
	\begin{equation}
		\boldsymbol{w}^*_n = \argmax_{\boldsymbol{w}\, | \, \|\boldsymbol{w}\|_2^2 = \nu^*_n} \Psi(\boldsymbol{w}), \, n \in \{1,2\},
	\end{equation}
	with probabilities $p^*_{\boldsymbol{w}}(\boldsymbol{w}_n^*) = p^*_{\nu}(\nu_n^*)$, $n \in \{1,2\}$, respectively.
	This concludes the proof.		
	\end{appendices}

	\bibliographystyle{IEEEtran}
	\bibliography{Final}

\end{document}